\documentclass[%
 reprint,
 amsmath,amssymb,
 aps,
prx,
]{revtex4-2}

\usepackage{graphicx}
\usepackage{dsfont}
\usepackage{xcolor}
\usepackage{comment}
\usepackage{dcolumn}
\usepackage{bm}

\usepackage{thm-restate}
\usepackage{amsmath}
\usepackage{amssymb}
\usepackage{mathtools}
\usepackage{amsthm}
\usepackage{algorithm}%
\usepackage{algpseudocode}%
\usepackage{float}

\algrenewcommand\algorithmicrequire{\textbf{Input:}}
\algrenewcommand\algorithmicensure{\textbf{Output:}}

\def \Tr {{\rm Tr}}
\def \Lhat {\widehat{L}}
\def \Ocal{\mathcal{O}}
\def \Ncal {\mathcal{N}}
\def \Hcal {\mathcal{H}}

\def \Ycal {\mathcal{Y}}

\def \Dcal {\mathcal{D}}
\def \Rcal {\mathcal{R}}
\def \Fcal {\mathcal{F}}
\def \Acal {\mathcal{A}}
\def \Ncal {\mathcal{N}}
\def \Ccal {\mathcal{C}}
\def \Gcal {\mathcal{G}}
\def \Acal {\mathcal{A}}
\def \xbf {\mathbf{x}}
\def \sbf {\boldsymbol{\sigma}}
\def \fbf {\mathbf{f}}
\def \Ebb {\mathbb{E}}
\def \Ibb {\mathbb{I}}

\def \rob {\mathcal{S}^{Q/C}_{r,p,\epsilon}}

\newtheorem{theorem}{Theorem}

\newtheorem{proposition}{Proposition}

\newtheorem{definition}{Definition}

\DeclareMathOperator{\sign}{\text{sign}}

\begin{document}

\preprint{APS/123-QED}

\title{On the Generalization of Adversarially Trained Quantum Classifiers}

\author{Petros Georgiou}
\email{pxg402@student.bham.ac.uk}
\author{Aaron Mark Thomas}
\author{Sharu Theresa Jose}

\affiliation{%
 Department of Computer Science, \\University of Birmingham, UK
}%

\author{Osvaldo Simeone}
\affiliation{
 KCLIP Lab\\Centre for Intelligent Information Processing Systems (CIIPS)\\Department of Engineering, King's College London, UK
}%

\date{\today}

\begin{abstract}
Quantum classifiers are vulnerable to adversarial attacks that manipulate their input classical or quantum data. A promising countermeasure is adversarial training, where quantum classifiers are trained by using an attack-aware, adversarial loss function. This work establishes novel bounds on the generalization error of adversarially trained quantum classifiers when tested in the presence of perturbation-constrained adversaries. The bounds quantify the excess generalization error incurred to ensure robustness to adversarial attacks as scaling with the training sample size $m$ as $1/\sqrt{m}$, while yielding insights into the impact of the quantum embedding. 
For quantum binary classifiers employing \textit{rotation embedding}, we find that, in the presence of adversarial attacks on classical inputs $\xbf$, the increase in sample complexity due to adversarial training over conventional training vanishes in the limit of high dimensional inputs $\xbf$. In contrast, when the adversary can directly attack the quantum state $\rho(\xbf)$ encoding the input $\xbf$,  the excess generalization error depends on the choice of embedding only through its Hilbert space dimension. The results are also extended to multi-class classifiers. We validate our theoretical findings with numerical experiments.
\end{abstract}

\maketitle


\section{Introduction}\label{sec:introduction} \vspace{-0.2cm}
\textbf{Context and Motivation:} Quantum Machine Learning (QML) aims to leverage quantum computing capabilities to outperform classical ML techniques \cite{Biamonte:QML, riste2017demonstration}. Recent studies have highlighted limitations of QML including difficulties in training unstructured QML models \cite{mcclean2018barren,wang2021noise} and the classical simulability of some structured QML models \cite{bermejo2024quantum}. Another concern with QML models is the fact that, similar to their classical counterparts, QML models are susceptible to adversarial attacks \cite{Lu_2020,ren2022experimental,west2023benchmarking}. For instance, a quantum classifier utilizing superconducting qubits to classify MRI images, achieving  a test accuracy of $99\%$, was found to be easily deceived by minor adversarial perturbations \cite{ren2022experimental}. This vulnerability poses another challenge on the way to realizing  quantum advantages. 
\begin{figure}
    \centering
    \includegraphics[width=0.96\linewidth]{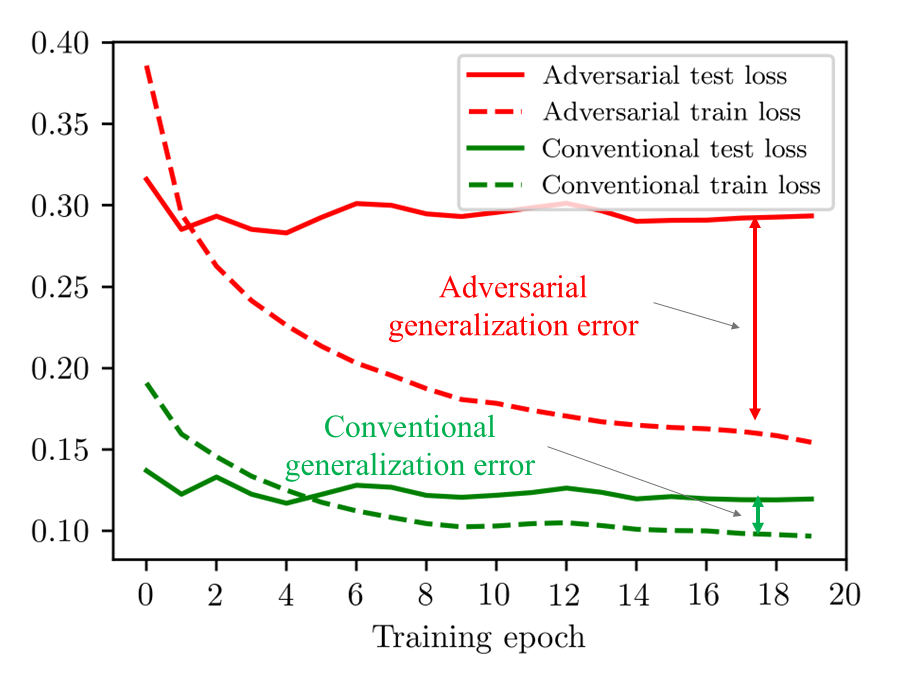}
    \caption{
    (Adversarial) Generalization error (Test error$-$Train error) plotted against number of training samples for an \emph{adversarially trained} angle embedding based classifier. Despite adversarial training, the adversarial generalization error remains significant. Details about the dataset and quantum classifier used can be found in section \ref{sec:embeddings}.}
    \label{fig:num_samples}
\end{figure}
To address this problem, recent works \cite{Du_2021,gong2024enhancing} have explored efficient strategies to defend quantum classifiers against adversarial attacks, with \textit{adversarial training} emerging as a promising strategy \cite{Lu_2020}.   Adversarial training replaces the standard classification loss with an attack-aware  \textit{adversarial loss}, accounting for the worst-case effect of adversarial perturbation of the input data.
This results in a min-max optimization problem with the classifier attempting to minimize the worst-case adversarial loss. 

In classical machine learning models it has been observed that adversarially trained classifiers have desirable training performance but a poor performance on test datasets \cite{NEURIPS2018_f708f064}. For example, a Resnet model adversarially trained on CIFAR10 dataset achieved a training accuracy of 96\%, but a test accuracy of only 47\% when evaluated on similarly perturbed data \cite{madry2019deeplearningmodelsresistant}. Theoretical studies have also confirmed that an increase in sample complexity is required to guarantee generalization with adversarial training \cite{yin2019rademacher,xiao2022adversarial}. 

The significant generalization gap of adversarially trained classifiers can also be observed in quantum classifiers. As an example, Figure~\ref{fig:num_samples} shows the adversarial training and test losses -- i.e., the training and test losses evaluated in the presence of an adversary -- incurred during adversarial training of a QML model over the training epochs. The figure also shows the conventional training and test losses of the same model. As observed in the figure, the generalization error -- i.e., the  gap between the training and test losses --  tends to be larger in the presence of an adversary. This is the case even though the QML model was trained using an adversarial loss.

While existing works analyze the empirical merit of adversarial training against various attacks \cite{Lu_2020,montalbano2024quantumadversariallearningkernel, ren2022experimental,west2023benchmarking, maouaki2025qfalquantumfederatedadversarial}, a theoretical understanding of the generalization performance of adversarially trained quantum classifiers remains limited. Addressing this knowledge gap is the focus of this work.

\textbf{Quantum Classifiers and Quantum Embeddings}: To elaborate, consider a quantum binary classifier that takes as  input a $d$-dimensional feature vector $\xbf$, embedding it into a $d_H$-dimensional  quantum state  $\rho(\xbf)$. Classification is carried out based on the sign of the function $f(\xbf)=\Tr(A\rho(\xbf))$, which corresponds to the expected value of a suitable observable $A\in\Acal$, with $\Tr(\cdot)$ denoting the trace. The set of observables $\Acal$ determines the model class, and the specific observable $A$ deployed for classification is identified based on training data.

The choice of the quantum embedding $\rho(\xbf)$ is critical for the performance and learnability of the quantum classifier \cite{schuld2021machine,simeone2022introduction,caro2021encoding, banchi2021generalization}.
\begin{figure}
    \centering
 \includegraphics[width=1.05\linewidth]{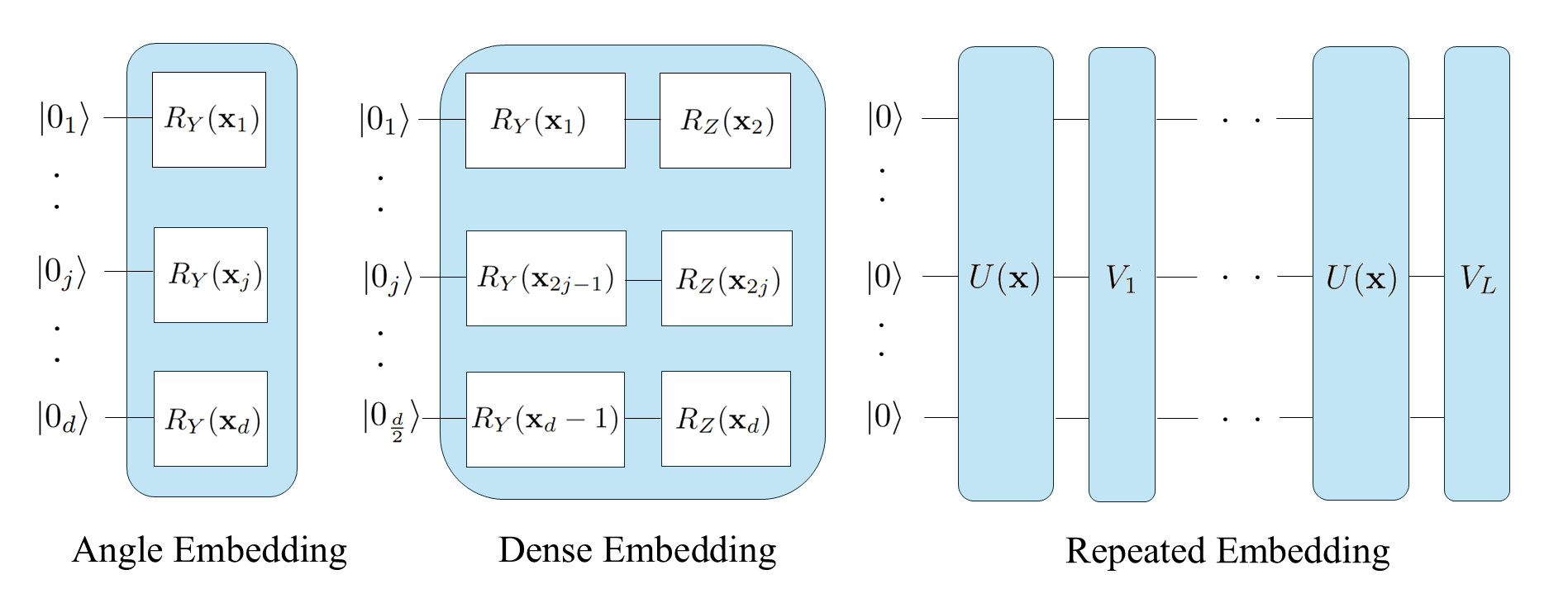}
    \includegraphics[width=1.05\linewidth]{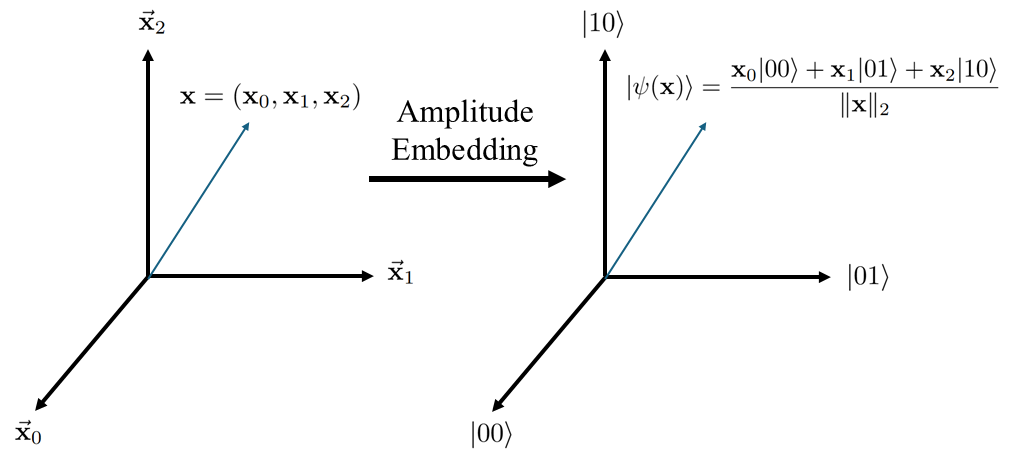}
    \caption{\small{The embedding $\rho(\xbf)$ plays a key role in determining the performance of a quantum classifier \cite{schuld2021machine,simeone2022introduction}. This paper studies the learning requirements for two main classes of embeddings in the presence of adversarial attacks. (Top) Rotation embedding schemes, including angle rotation embedding, dense rotation embedding, and their repeated variants require a number of qubits that scale linearly with the dimension $d$ of the input vector $\mathbf{x}$. (Bottom) Amplitude embedding, which can encode a vector $\mathbf{x}$ of dimension $d$ into a quantum state whose number of qubits scales logarithmically with $d$. In this particular case we embed a $3$-dimensional input $\xbf$ into a $2$-qubit quantum state.}}
    \label{fig:embeddings} \vspace{-0.5cm}
\end{figure}
As illustrated in Figure~\ref{fig:embeddings}, we distinguish two main types of embeddings. The first type encompasses \emph{rotation}-based methods, in which information is encoded via the rotation angles of qubit gates. These methods, encompassing angle embedding, dense embedding, and their repeated variants \cite{dowling2024adversarial},  require a number of qubits that scales \emph{linearly} with the data dimension $d$, and thus the Hilbert space dimension $d_H$ grows exponentially with $d$. Despite their significant requirements in terms of number of qubits, these embeddings are widely used in practice as the corresponding circuits have constant depth \cite{schuld2021machine, LaRose_2020}.

In contrast, with \emph{amplitude encoding}, the information vector $\mathbf{x}$ directly modulates the quantum amplitudes of the quantum state $\rho(\mathbf{x})$. Accordingly, in order to embed a vector $\xbf$ of dimension $d$ into a quantum state, amplitude embedding requires a number of qubits that grows only logarithmically with $d$, and thus the dimension of the Hilbert space $d_H$ scales linearly with $d$. On the flip side, the circuit implementing amplitude embedding generally requires a number of quantum gates that is exponential in the dimension $d$ \cite{schuld2021machine}.

 The impact of the choice of quantum embedding on the generalization of conventionally trained classifiers has been studied in \cite{caro2021encoding}. This reference quantifies the generalization error via the number of encoding gates used in the circuit. In particular, for an $L$-layer angle encoding based classifier, the generalization error is specifically shown to scale as $\Ocal(L^d)$. To the best of our knowledge, generalization for amplitude encoding has not been explicitly addressed in the literature.


\textbf{Research Questions:}  Different from classical ML models, adversaries acting against QML models can attack data either in the classical  space, perturbing the input $\xbf$ within an $l_p$-ball of some radius $\epsilon$, or in the quantum state space, perturbing the quantum state $\rho(\xbf)$ within a $p$-Schatten norm ball of $\epsilon$-radius \cite{dowling2024adversarial}. Accordingly, we study the \textit{adversarial generalization error} -- defined as the difference between the adversarial training loss and the test loss -- for adversarially trained QML models with angle or amplitude embeddings under classical or quantum attacks.  Specifically, we address the following key questions:\vspace{-0.2cm}\begin{itemize}
\item How is the sample complexity of a QML model  affected by the presence of an attacker in the ideal case in which the training process is aware of the limitations of the attacker?\vspace{-0.2cm} \item How does the choice of the  quantum embedding  affect adversarial generalization? \vspace{-0.2cm} \item What is the impact of the type of attack, classical or quantum, on adversarial generalization? \end{itemize}
We answer these questions by providing novel probably approximately correct (PAC) bounds on the adversarial generalization error via the \emph{adversarial Rademacher complexity} (ARC). This extends the standard Rademacher complexity (RC)-based PAC bounds presented in \cite{mohri2012foundations} to the adversarial setting.

\textbf{Technical Challenges:} Aside from our conference paper \cite{georgiou_ISIT24}, the adversarial generalization of quantum classifiers has yet to be explored at a theoretical level. This is in contrast to extensive investigations into classical classifiers \cite{khim2018adversarial, yin2019rademacher,awasthi2020adversarial}.  
In particular, for linear classifiers, the ARC was previously analyzed in \cite{yin2019rademacher} by explicitly solving the inner maximization over adversarial perturbations. Although the quantum classifiers discussed here are linear in the embedded quantum state $\rho(\xbf)$, they can exhibit highly non-linear behavior with respect to the input $\xbf$ due to the choice of quantum embedding. This complicates the analysis of ARC under classical attacks. For quantum attacks, this problem is compounded by the constraint that adversarial perturbations yield valid quantum states.

\begin{table}[h!]
\begin{footnotesize}
\resizebox{0.5\textwidth}{!}{%
\begin{tabular}{ |c|c|c|} 
   \hline Attack & Embedding & scaled excess RC\\
   \hline Classical &
   Amplitude & $2^{1+\lceil\log_2d\rceil}b\min\Big\{\frac{\epsilon b\max\{1, d^{1/2-1/p}\}}{\min_{x\in\mathcal{X}}\lVert x\rVert_{2}\sqrt{m}},\frac{1}{\sqrt{m}}\Big\}$\\ \hline
   Classical & $L$-layer Angle  &$\frac{(2\epsilon)^{d}Lbd^{-d/p}}{\sqrt{m}}$\\\hline Classical & $L$-layer Dense &$\frac{(2\sqrt{2}\epsilon)^{d/2}Lb\max\{d^{-d/4},d^{-d/2p}\}}{\sqrt{m}}$\\
  \hline Quantum &Any &$\frac{\epsilon bd_H\max\{1,d_H^{1-1/p-1/r}\}}{\sqrt{m}}$ \\\hline
\end{tabular}}
   \caption{\small{ Behavior of the scaled excess Rademacher complexity  $\rob/\sqrt{m}$  according to the bounds derived in this paper on the ARC (\ref{eq:first}) for classical/quantum attacks and for different types of embeddings ($d$ is the data dimension, $m$ the number of training samples; $(p,\epsilon)$ denote the parameters of the adversarial attack; $b$ is the bound on the observable $r$-Schatten norm; and $d_H$ is the Hilbert space dimension).}} \label{tab:comparison}
  \end{footnotesize}
  \vspace{-0.3cm}
   \end{table}
{\textbf{Main Contributions}:} Our previous conference paper \cite{georgiou_ISIT24} addressed the adversarial generalization error under quantum attacks for a subset of the possible embeddings $\rho(\xbf)$ for which direct evaluation of the ARC is possible. In contrast, this more comprehensive study investigates the adversarial generalization error for both quantum and classical attacks using  \emph{covering number-based} analysis \cite{vershynin2018high}. Specifically, we consider a classifier function class defined by a set of observables $\Acal_r$ constrained via the $r$-Schatten norm.  We also define a $(p,\epsilon)$-adversary capable of perturbing the input $\xbf$ within an $\ell_p$ norm ball of radius $\epsilon>0$ for classical attacks, or perturbing the input quantum state $\rho(\xbf)$ within a $p$-Schatten norm ball of radius $\epsilon>0$ for quantum attacks.  We show that the ARC of binary quantum classifiers can be upper bounded as the sum of the standard RC of the classifier function class and of an additional term quantifying the \textit{excess Rademacher complexity} due to the presence of an adversary, i.e.,
\begin{align}
    {\rm ARC} \leq {\rm RC}+ \mathcal{O}\left(\frac{\rob}{\sqrt{m}}\right). \label{eq:first}
\end{align}
We derive novel bounds on the excess term $\rob$ as a function of the norm of the observable $r$, and the power of the adversary $(p,\epsilon)$ for classical ($C$) and quantum ($Q$) attacks. A larger value of the excess Rademacher complexity $\rob$ indicates a larger \emph{sample complexity} for adversarial training as compared to conventional training.


Through this analysis, we investigate the impact of the quantum embedding $\rho(\mathbf{x})$ on both RC and on the excess Rademacher complexity $\rob$. Specifically, as shown in  Figure~\ref{fig:embeddings},  we consider  rotation based embeddings, along with  extensions via repeated embedding and  data re-uploading  \cite{LaRose_2020}, as well as amplitude embedding. Our main conclusions are summarized in Table~\ref{tab:comparison} and can be described as follows.
\begin{itemize}\vspace{-0.1cm}
    \item 
    With \textit{rotation embedding} and \emph{classical} adversarial attacks, the sample complexity of adversarial training tends to that of standard training in the limit $d\gg e^{\mathcal{O}(p\ln(2\sqrt{2}\epsilon))}$. In fact, in this limit, the ARC approaches the standard RC asymptotically, with the excess Rademacher complexity $\rob$ decaying at least exponentially with the data dimension $d$.
 \item With amplitude embedding and  \emph{classical} adversarial attacks, quantum classifiers employing \textit{amplitude embedding} the excess Rademacher complexity scales at least linearly with the dimension $d$ of the input data $\xbf$.
\item  With \emph{quantum} adversarial attacks we identify that the excess RC only depends on the choice of embedding through its Hilbert space dimension $d_H$. In the case of sufficiently noisy embeddings we are able to show that ARC $\geq$ RC. Furthermore, for the same class of embeddings we provide tighter bounds on the excess RC.
\end{itemize}
\vspace{-0.2cm}
\textbf{Related Works}: The generalization performance of standard, non-adversarial QML models has been extensively studied, including approaches based on complexity measures such as the RC \cite{caro2021encoding, caro2022generalization, bu2022statistical, banchi2021generalization, banchi2023statistical,jose2023transfer}, and on algorithmic stability analyses \cite{yang2025stabilitygeneralizationquantumneural, zhu2025optimizerdependentgeneralizationboundquantum}. In contrast, works focusing on adversarial generalization of quantum classifiers are scarce. Recent work \cite{berberich2023training} studies the generalization of quantum classifiers with trainable quantum embeddings under conventional classification loss. This is done via Lipschitz bounds that also determine the adversarial robustness. In contrast, our focus is on analyzing the generalization of quantum classifiers that are adversarially trained. 

The work most closely related to ours is our prior conference paper \cite{georgiou_ISIT24}, which studies adversarial generalization in the presence of a quantum adversary. This work derives an upper bound on the adversarial generalization error considering a restricted class of quantum embeddings for which the minimum eigenvalue of the embedding is greater than $\epsilon$. In contrast, we address the impact of the choice of the quantum embeddings on the adversarial generalization error.

Other existing works in adversarial QML primarily aim at developing efficient adversarial defense strategies, or on proving certifiable adversarial robustness guarantees. In the former category, reference  \cite{Du_2021,huang2023enhancing} explores
leveraging quantum gate noise layers (e.g., depolarization noise) for adversarial defence, while reference \cite{gong2024enhancing} investigates using randomized encodings (see also \cite{guan2021robustness, guan2022verifying, huang2023certified}). Furthermore, the works \cite{liu_vulnerability, liao2021robust, dowling2024adversarial} investigate the robustness of QML models as a function of size of the quantum system used, indicating at a trade-off between system size and robustness, a potential threat to achieving quantum advantage.
\vspace{-0.2cm}
\section{Quantum Classifiers}
\vspace{-0.2cm}
Let $(\xbf,y)$ denote the (classical) data tuple of feature vector $\xbf=(x_1,\hdots,x_d) \in\mathcal{X}\subseteq \mathbb{R}^d$ and output label $y \in \Ycal$  where $\Ycal=\{+1,-1\}$ for binary classification and $\Ycal=\{1,\hdots,K\}$ for $K$-class classification. For quantum data, the input feature is a quantum state $\rho(\xbf)$ indexed by the parameter vector $\xbf$ of the state, and $y$ denotes the corresponding label.

A quantum classifier is described by a mapping $g:\mathcal{X} \rightarrow \Ycal$ from the input feature space to the label space. When the data are  classical,  the classifier first maps the input vector $\xbf \in \mathcal{X}$ to a quantum state $\rho(\xbf)$ in a $d_H$-dimensional Hilbert space via the quantum embedding  $ \xbf \mapsto \rho(\xbf)$. This step can be skipped if the data are  quantum.  Note that the quantum state $\rho(\xbf)$ is a \textit{density matrix}, i.e., a unit-trace,  positive semi-definite, $d_H \times d_H$ Hermitian matrix. 

The classifier  measures the embedded quantum state using a suitable Hermitian observable $A$. Throughout this work, we consider a \emph{model class} of classifiers defined by the set of observables 
\begin{align}
 \Acal_r&=\{A:A=A^{\dagger}, \Vert A \Vert_r\leq b\},\label{eq:Ocal}
\end{align} i.e.,  the set of $r$-Schatten norm constrained Hermitian matrices of dimension $d_H \times d_H$. For $r \in [1,\infty)$, the Schatten norm of observable $A$ is defined as $\Vert A \Vert_r= ({\rm Tr}(|A|^r))^{1/r}$ with $|A|=\sqrt{A^{\dag}A}$, which in the limiting case of $r=\infty$, evaluates as $\Vert A \Vert_{\infty}=\max\{\lambda(|A|)\}$ with $\{\lambda(|A|)\}$ denoting the set of eigenvalues of $|A|$.

The quantum embedding and the measurement define a  binary quantum classifier as \cite{gyurik2023structural}
\begin{align}
    g(\xbf) & = {\rm sgn}(f(\xbf)), \quad \mbox{with} \quad f(\xbf) = {\rm Tr}(A\rho(\xbf)), \label{eq:quantumclassifier}
\end{align} where the real-valued function $f(\xbf)$ computes the expectation of the Hermitian observable $A$ with respect to the quantum state $\rho(\xbf)$, with ${\rm Tr}(\cdot)$ denoting the trace operation. This is then post-processed via the function  ${\rm sgn}(a)$, which takes value $+1$ when $a>0$ and takes $-1$ when $a<0$, to yield the class label.

For fixed quantum embedding $\xbf \mapsto \rho(\xbf)$, the class $\Gcal_r$ of binary quantum classifiers can be then defined as  $\Gcal_r=\{\xbf \mapsto g(\xbf)={\rm sgn}(f(\xbf)): f\in \Fcal_r\}$, where \begin{align}
    \Fcal_r&=\{\xbf \mapsto f(\xbf)= {\rm Tr}( A\rho(\xbf)): A \in \Acal_r\} .\label{eq:Fcal} \end{align}
More generally, a $K$-class quantum classifier can be defined via $K$ observables $\{A_k\}_{k=1}^K$ as
\begin{align}
    g(\xbf) = \arg \max_{k \in [K]}f_k(\xbf), \quad \mbox{with} \hspace{0.1cm} f_k(\xbf)=\Tr(A_k \rho(\xbf)) \label{eq:multiclass_quantumclassifiers}
\end{align} for $k \in [K]=\{1,\hdots,K\}$. The resulting class of quantum classifiers is $\Gcal_r=\{ \xbf \mapsto g(\xbf)=\arg \max_{k \in [K]}f_k(\xbf): f_k \in \Fcal_r \hspace{0.1cm} \text{for all } k \in [K]\}$.

\emph{Remark 1:} The definition of quantum classifiers above assumes a general class $\Acal_r$ of norm-bounded observables. In practice, the observable $A$ is implemented by leveraging parameterized unitary gates as in quantum neural networks. For instance, in binary classification, the observable $A$ takes the form $A=UZ_1U^{\dag}$ where $U$ denotes a tunable, parameterized unitary gate and $Z_1$ denotes a (fixed) Pauli-Z measurement on the first qubit. Such parameterized observables form a subset of the set $\Acal_r$. \qed 

Before defining the adversarial generalization error of the quantum classifiers, following Section~1 and Figure~\ref{fig:embeddings}, we first discuss conventional methods for embedding classical data into quantum states. As we will see, the choice of the embedding scheme is crucial for adversarial generalization.
\vspace{-0.2cm}
\subsection{Quantum Embedding Schemes}\label{sec:quantumembedding}
We now describe the quantum embeddings $\xbf \mapsto \rho(\xbf)$ that will be studied in this work. Pure-state embeddings are of the form $\rho(\xbf)=\vert \psi(\xbf) \rangle \langle \psi(\xbf) \vert$, where $\vert \psi(\xbf) \rangle \in \mathbb{C}^{d_H}$ is a $d_H$-dimensional complex vector with unit norm.

\textbf{Amplitude embedding} describes the mapping 
\vspace{-0.1cm}
\begin{align}
\xbf \mapsto |\psi(\xbf)\rangle=\sum_{i=1}^d\frac{x_i}{\lVert\xbf\rVert_2}|i\rangle,
\end{align}that encodes the normalized feature entries to the amplitudes of computational basis states $\vert i \rangle$ for $i \in [d]$. Implementing this embedding requires $n= \lceil\log_2 d\rceil$  qubits, resulting in a $d_H=2^n= 2^{\lceil\log_2 d\rceil}$-dimensional Hilbert space \cite{schuld2021supervised}.

\textbf{Angle rotation embedding} uses $n=d$ qubits to encode the $d$-dimensional input $\xbf$. Specifically, the $i$-th entry of the feature vector $\xbf$ is encoded into the $i$-th qubit via a Pauli rotation. Using Pauli-Y rotations, this is achieved via the application of a $d$-qubit unitary gate $U(\xbf) = \bigotimes_{j=1}^d e^{-ix_j\sigma_Y}$, with $\sigma_Y$ denoting the Pauli-Y matrix, on an initial state $\vert 0 \rangle$ as $\vert \psi(\xbf)\rangle = U(\xbf) \vert 0 \rangle$. The resulting embedded quantum state is obtained as
\begin{align*}
   \xbf \mapsto \vert \psi(\xbf) \rangle= U(\xbf) \vert 0 \rangle.
\end{align*}
{\textbf{Dense rotation embedding}} uses $n=d/2$ qubits to encode two feature entries into each qubit via the unitary matrix $U(\xbf)=\bigotimes_{j=1}^{d/2} e^{-ix_{2j}\sigma_Z}e^{-ix_{2j-1}\sigma_Y}$ with $\sigma_Z$ denoting Pauli-Z matrix \cite{dowling2024adversarial}. The resulting embedded quantum state $\rho(\xbf)$ can be obtained as in angle rotation embedding.

Finally, \textbf{$L$-layer rotation embedding} schemes obtain the quantum state $\vert \psi(\xbf) \rangle$ by applying the rotation unitary $U(\xbf)$ followed by a fixed unitary $V^{(l)}$ $L$-times. This results in \begin{align} \vert \psi(\xbf) \rangle =  V^{(L)} U(\xbf) \hdots V^{(2)} U(\xbf) V^{(1)} U(\xbf) \vert 0 \rangle.\end{align}  Such repeated encoding schemes that facilitate data re-uploading are known to yield highly expressive quantum models \cite{schuld2021effect}.

Having defined quantum embeddings and set $\Acal_r$ of observables that define the quantum classifier, we now define the generalization error of quantum classifiers.

\section{Generalization Error: Learning-Theoretic Preliminaries}
Let each data tuple $(\xbf,y)$ be sampled from an unknown data distribution $P(\xbf,y)$.
A learner observes a data set $\mathcal{D}=\{(\xbf_1,y_1),\hdots, (\xbf_m,y_m)\}$ of $m$ examples, generated i.i.d according to the distribution $P(\xbf,y)$, and uses it to learn a quantum classifier $g: \mathcal{X} \rightarrow \Ycal$. 
\vspace{-0.1cm}
\subsection{Non-Adversarial Generalization Error} In the non-adversarial setting,
the performance of the quantum classifier $g$ on a data point $(\xbf,y)$ is measured using a non-negative loss function $\ell(g,\rho(\xbf),y)$. The goal of the learner is to find the classifier $g \in \Gcal_r$ that minimizes the \textit{population risk},
\begin{align}
    L(g) =\Ebb_{P(\xbf,y)}[\ell(g,\rho(\xbf),y)], \nonumber \end{align}  which is the expected loss over test data tuples drawn from $P(\xbf,y)$. Since this cannot be evaluated due to the unknown data distribution $P(\xbf,y)$, the learner instead minimizes the \textit{empirical} risk,
\begin{align}
    \Lhat(g) =\frac{1}{m}\sum_{i=1}^m \ell(g,\rho(\xbf_i),y_i),\end{align}
evaluated on the observed data set $\Dcal$. The difference between the population risk and empirical risk is the \textit{generalization error},
\begin{align}
    G(g) = L(g)-\Lhat(g),
\end{align} of the classifier $g$.

Understanding the generalization error of ML models is the central question of statistical learning theory \cite{shalev2014understanding} and there exists different approaches to quantify it. In this work, we resort to the \emph{probably approximately correct} (PAC) framework that characterizes the generalization error via RC as in the theorem below.
\begin{theorem}\label{thm:uniformconvergence}\cite{shalev2014understanding}
Assume that the loss function $\ell(g,\rho(\xbf),y)$ is $[0,B]$-bounded. Then, with probability at least $1-\delta$, for $\delta \in (0,1)$, with respect to the random draws of dataset $\Dcal$, the generalization error $G(g)$ of any classifier $g \in \Gcal_r$ can be upper bounded as
\begin{align}
    G(g) &\leq 2\Rcal(\ell \circ \Gcal_r)+3B \sqrt{\frac{1}{2m} \log \frac{2}{\delta}}, \quad \mbox{where} \\
  \Rcal(\ell \circ \Gcal_r)&=  \Ebb_{\boldsymbol{\sigma}}\Bigl[\sup_{g \in \mathcal{G}_r}\frac{1}{m} \sum_{i=1}^m \sigma_i \ell(g,\rho(\xbf_i),y_i) \Bigr] \label{eq:non-adversarial_Rademacher}
\end{align} is the empirical RC of class $\Gcal_r$ evaluated on dataset $\Dcal$, with $\boldsymbol{\sigma}=(\sigma_1,\hdots,\sigma_m)$ denoting the vector of i.i.d Rademacher variables $\sigma_i \in \{+1,-1\}$ taking values $\pm 1$  with equal probability.
\end{theorem}
Theorem~\ref{thm:uniformconvergence} shows that the Rademacher complexity of the class $\Gcal_r$ of classifiers can bound the generalization error with high probability.
\vspace{-0.2cm}
\subsection{Adversarial Generalization Error}
We now consider the quantum classification problem in an adversarial setting. Specifically, we consider \textit{white-box adversaries} that have access to the classifier $g(\cdot)$ as well as the loss function $\ell(g,\rho(\xbf),y)$. Depending on the data space that the adversary attacks, we identify two distinct types of adversarial attacks: classical and quantum adversarial attacks. These are depicted in Figure~\ref{fig:adversaries}.
\begin{figure}
    \centering
    \includegraphics[width=1\linewidth]{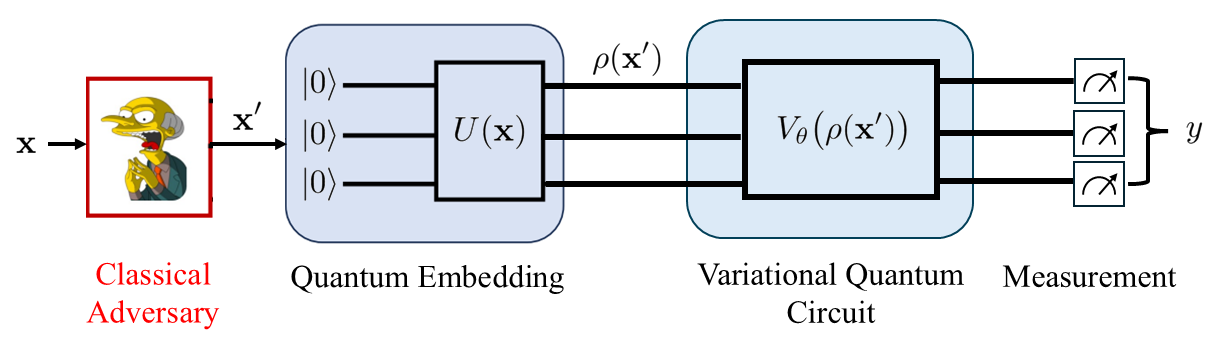}
    \includegraphics[width=1\linewidth]{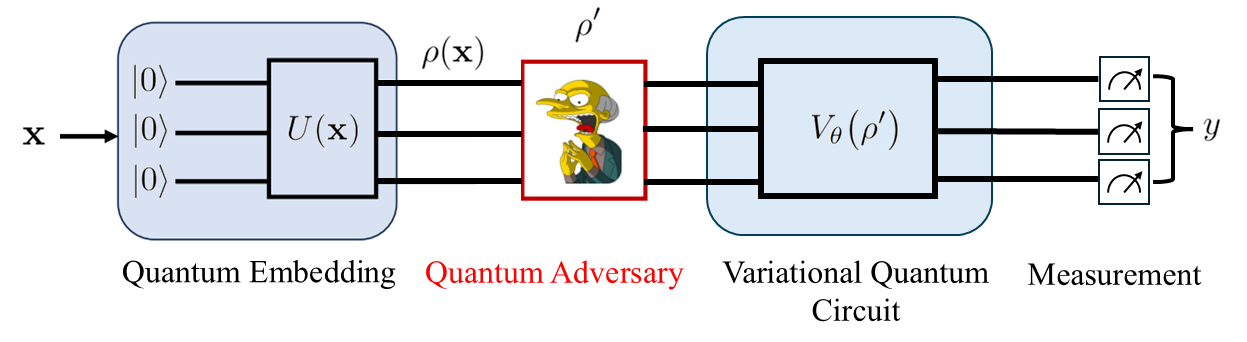}
    \caption{Illustration of a quantum classifier being subject to classical (top) and quantum (bottom) adversarial attacks.}
    \label{fig:adversaries}
\end{figure}

A \textit{classical adversary} perturbs the input feature vector $\xbf \in \mathcal{X}$ to an adversarial feature $\xbf' \in \mathcal{X}$ satisfying the constraint 
\begin{align}
    \Vert \xbf-\xbf' \Vert_p \leq \epsilon,
\end{align} such that it maximizes the loss $\ell(g,\rho(\xbf'),y)$ incurred by the classifier $g$  \cite{dowling2024adversarial}. We call the above attack a \textit{$(p,\epsilon)$-classical attack} since the adversary performs an $l_p$-norm perturbation of the clean input vector within a budget of $\epsilon>0$.  

In contrast, a \textit{quantum adversary} directly perturbs the quantum state $\rho(\xbf)$ to an adversarial quantum state $\rho'$, that satisfies the $p$-Schatten norm constraint 
\begin{align}
\Vert \rho(\xbf)-\rho' \Vert_p \leq \epsilon,
\end{align} 
to maximize the classifier's loss. We call the resulting attack a \textit{$(p,\epsilon)$-quantum adversarial attack}, which assumes that the  adversary has direct access to the quantum state $\rho(\xbf)$.

The attacks defined above incur the following \textit{adversarial loss}  for $(p,\epsilon)$-classical attack,
\begin{align}
    \ell^C_{r,p,\epsilon}(g,\xbf,y)&=\max_{\xbf': \Vert \xbf-\xbf'\Vert_p \leq \epsilon} \ell(g,\rho(\xbf'),y), \quad \mbox{and}\\
    \ell^Q_{r,p,\epsilon}(g,\xbf,y)&=\max_{\rho': \Vert \rho(\xbf)- \rho'\Vert_p \leq \epsilon} \ell(g,\rho',y),
\end{align} for the $(p,\epsilon)$-quantum attack. The superscript $C/Q$ on the adversarial loss denotes if it is under classical or quantum adversary.

Assuming that the learner is aware of the presence of the adversary, their goal is to find the classifier $g \in \Gcal_r$ that minimizes the classical/quantum \textit{adversarial population risk}
\begin{align}
    L^{C/Q}_{r,p,\epsilon}(g)=\Ebb_{P(\xbf,y)}[\ell_{r,p,\epsilon}^{C/Q}(g,\xbf,y)]. \nonumber
\end{align} Since the adversarial population risk cannot be evaluated, the learner instead resorts to \textit{adversarial training} via minimizing the \textit{adversarial empirical risk},
\begin{align}
    \Lhat^{C/Q}_{r,p,\epsilon}(g)=\frac{1}{m}\sum_{i=1}^m  \ell^{C/Q}_{r,p,\epsilon}(g,\xbf_i,y_i). \nonumber
\end{align} The $(p,\epsilon)$-classical/quantum \textit{adversarial generalization error} is then defined as $$G_{r,p,\epsilon}^{C/Q}(g)=L^{C/Q}_{r,p,\epsilon}(g)-\Lhat^{C/Q}_{r,p,\epsilon}(g),$$  the difference between the adversarial population and adversarial empirical risks.

To characterize the adversarial generalization error, one can straightforwardly extend Theorem~\ref{thm:uniformconvergence} to adversarial setting. This results in an upper bound on the  adversarial generalization error $G^{C/Q}_{r,p,\epsilon}(g)$ under classical or quantum adversary via the ARC, 
\begin{align}
  \Rcal(\ell^{C/Q}_{r,p,\epsilon} \circ \Gcal_r)=  \Ebb_{\boldsymbol{\sigma}}\Biggl[\sup_{g \in \mathcal{G}_r}\frac{1}{m} \sum_{i=1}^m \sigma_i \ell^{C/Q}_{r,p,\epsilon}(g,\xbf_i,y_i) \Biggr]  \label{eq:adversarial_Rademacher}.
\end{align} This poses the challenge of analyzing the  ARC $\Rcal(\ell^{C/Q}_{r,p,\epsilon} \circ \Gcal_r)$ of the class $\Gcal_r$ of quantum binary or multi-class classifiers. 
To this end, we resort to covering number-based analysis \cite{vershynin2018high, mohri2012foundations}, deriving the novel bounds on the ARC of binary classifiers defined in the next section.
\section{Adversarial Rademacher Complexity of Binary Classifiers}
In this section, we study the adversarial generalization of binary quantum classifiers of the form $g(\xbf)={\rm sgn}(f(\xbf))$ as defined in \eqref{eq:quantumclassifier}. 
Following \cite{yin2019rademacher} and \cite{awasthi2020adversarial},  we consider the following loss function
\begin{align}
    \ell(g,\rho(\xbf),y) = \phi(yf(\xbf)), \label{eq:binaryloss}
\end{align} where $\phi:\mathbb{R} \rightarrow [0,B]$ is a $[0,B]$-bounded, monotonically non-increasing, $\eta$-Lipschitz continuous function. We first bound the non-adversarial Rademacher complexity $\Rcal(\ell \circ \Gcal_r)$ of this class of classifiers.

\subsection{Non-adversarial Rademacher Complexity Bounds}
In the non-adversarial setting, the Rademacher complexity \eqref{eq:non-adversarial_Rademacher} of the class $\Gcal_r$ of binary classifiers $g(\xbf)={\rm sgn}(f(\xbf))$ under loss function \eqref{eq:binaryloss} evaluates as
\begin{align}\Rcal(\ell \circ \Gcal_r)=\Ebb_{\boldsymbol{\sigma}} \Biggl[\sup_{f \in \Fcal_r} \frac{1}{m}\sum_{i=1}^m \sigma_i \phi(y_i f(\xbf_i)) \Biggr] \nonumber.
\end{align} Since $\phi(\cdot)$ is $\eta$-Lipschitz, $\Rcal(\ell \circ \Gcal_r)$ can be upper bounded using Talagrand's contraction inequality \cite{talagrand} as
\begin{align}
\Rcal(\ell \circ \Gcal) &\leq \eta \Rcal(\Fcal_r) \hspace{0.2cm} \mbox{with} \nonumber \\\Rcal(\Fcal_r)&=\Ebb_{\boldsymbol{\sigma}} \Biggl[\sup_{f \in \Fcal} \frac{1}{m}\sum_{i=1}^m \sigma_i f(\xbf_i) \Biggr], \label{eq:Rcal_Fcal}
\end{align}denoting the RC of the function class $\Fcal_r$ in \eqref{eq:Fcal}.
The following theorem upper bounds the RC $\Rcal(\Fcal_r)$ defined above.
\begin{theorem}\label{thm:binary_nonadversarial} 
For quantum binary classifiers, the empirical RC $\Rcal(\Fcal_r)$, evaluated on an $m$-sample dataset $\Dcal$, can be upper bounded as
\begin{align*}
\Rcal(\Fcal_r) \leq\frac{b}{m}\begin{cases}
\sqrt{2\ln(d_H)\Big\lVert\sum_{i=1}^m\rho(\bold{x}_i)^2\Big\rVert_\infty}&\hspace{-0.3cm},\;r=1\\\\
        B_{\frac{r}{r-1}}\lVert \sqrt{\sum_{i=1}^m\rho(\bold{x}_i)^2}\rVert_{\frac{r}{r-1}}&\hspace{-0.3cm},\;2\leq r<\infty\\\\\lVert\sqrt{\sum_{i=1}^m\rho(\mathbf{x}_i)^2}\rVert_1&\hspace{-0.3cm},\;r=\infty
    \end{cases}
\end{align*}
where $B_{\beta}=2^{-\frac{1}{4}}\sqrt{\frac{\pi \beta}{e}}$.
\end{theorem}
\begin{proof}
Proof can be found in Appendix~\ref{app:binary_nonadversarial}.
\end{proof} It is important to note that apart from the explicit Hilbert space dimensional dependence of these bounds (as seen in $r=1$ case), they also implicitly depend on $d_H$ through the norms of  $\sum_i\rho(\mathbf{x}_i)^2$. In most practical applications, observables are based on Pauli operators, which have bounded $r=\infty$ norm.

\subsection{Adversarial Rademacher Complexity Bounds} 
We start by analyzing the ARC \eqref{eq:adversarial_Rademacher} under classical attacks.
Note that since the function $\phi(\cdot)$ is monotonically non-increasing, the ARC  under classical attack can be equivalently written as 
\begin{align}
\Rcal(\ell^{C}_{r,p,\epsilon} \circ \Gcal_r)
& = 
\Ebb_{\boldsymbol{\sigma}}\Biggl[\sup_{f \in \Fcal_r} \frac{1}{m} \sum_{i=1}^m  \phi \Bigl(\min_{\xbf': \Vert \xbf-\xbf'\Vert_p \leq \epsilon}y_if(\xbf_i)\Bigr) \Biggr] \nonumber.
\end{align} Upper bounding via Talagrand's contraction inequality \cite{talagrand} yields,
\begin{align}
    \Rcal(\ell^{C}_{r,p,\epsilon} \circ \Gcal_r) & \leq \eta \Ebb_{\boldsymbol{\sigma}}\Biggl[\sup_{h \in \Fcal^{C}_{r,p,\epsilon}} \frac{1}{m}\sum_{i=1}^m \sigma_i h(\xbf_i,y_i)\Biggr] \nonumber \\&=\eta \Rcal(\Fcal_{r,p,\epsilon}^C),\label{eq:adversarial_equivlaent_binary} 
\end{align}
where $\Fcal_{r,p,\epsilon}^C$ is the adversarial function class  defined as
\begin{align*}
    \Fcal_{r,p,\epsilon}^C &=\Bigl\{(\xbf,y) \mapsto h(\xbf,y)= \hspace{-0.2cm}\min_{\xbf': \Vert \xbf-\xbf'\Vert_p \leq \epsilon}y f(\xbf'): f \in \Fcal_r\Bigr\}.
    \end{align*}
Similarly, under quantum attack, we get that $\Rcal(\ell^Q_{r,p,\epsilon} \circ \Gcal_r) \leq \eta \Rcal(\Fcal_{r,p,\epsilon}^Q)$ with
    \begin{align}\Fcal_{r,p,\epsilon}^Q =\Bigl\{(\xbf,y) \mapsto h(\xbf,y)&= \hspace{-0.2cm} \min_{\rho': \Vert \rho(\xbf)-\rho'\Vert_p \leq \epsilon} y\Tr(A\rho')\nonumber\\&: A \in \Acal_r\Bigr\}.\end{align}

The following theorem presents an upper bound on the  ARC $\Rcal(\Fcal_{r,p,\epsilon}^{C/Q})$ of the adversarial function class $\Fcal_{r,p,\epsilon}^{C/Q}$ via covering number-based analysis. The proof is deferred to Section~\ref{sec:proof}.


\begin{theorem}\label{thm:upperbound_binary}
   The ARC $\Rcal(\Fcal_{r,p,\epsilon}^{C/Q})$  of adversarial function class $\Fcal_{r,p,\epsilon}^{C/Q}$, evaluated on $m$-sample dataset $\Dcal$, under classical/quantum adversarial attack, with $p,r \geq 1$, can be upper bounded as
    \begin{align}
        &\Rcal(\Fcal_{r,p,\epsilon}^{C/Q}) \leq  \Rcal(\Fcal_r)+ \frac{b\rob J(r)}{\sqrt{m}},\label{eq:adversarial_binary_upperbound}
    \end{align}
    where $J(r)$ scales as $\Ocal(\sqrt{1+1/r})$, and
    \begin{align}
\mathcal{S}_{r, p,\epsilon}^{C}&=d_H\sup_{\xbf,\xbf':\Vert \xbf-\xbf'\Vert_p \leq \epsilon} \Vert \rho(\xbf) -\rho(\xbf')\Vert_{\frac{r}{r-1}}, \label{eq:classical_c}\\
\mathcal{S}_{r, p,\epsilon}^{Q}&=d_H\sup_{\rho(\bold{x}),\rho':\Vert \rho(\bold{x})-\rho'\Vert_p \leq \epsilon} \Vert \rho(\bold{x}) -\rho'\Vert_{\frac{r}{r-1}} \nonumber \\&\leq \epsilon\max\{d_H, d_H^{2-1/r-1/p}\}. \label{eq:quantum_c}
    \end{align}
\end{theorem} 
\begin{proof}
Proof can be found in Appendix~\ref{app:binary_adversary_upperbound}.
\end{proof}

The bound in \eqref{eq:adversarial_binary_upperbound} shows that $\Rcal(\Fcal_{r,p,\epsilon}^{C/Q})$ is upper bounded by the sum of the  RC $\Rcal(\Fcal_r)$ of non-adversarial function class $\Fcal_r$ and an additional term, the \textit{excess Rademacher complexity} due to adversarial training, scaling as $\mathcal{O}(b\rob/\sqrt{m})$. In the absence of adversarial attacks (i.e., $\epsilon=0$), the excess RC vanishes since $\rob=0$. For $\epsilon>0$ and under classical adversarial attack, the constant $\mathcal{S}_{r,p,\epsilon}^{C}$ in \eqref{eq:classical_c} is a measure of \textit{smoothness} of the quantum embedding $\xbf \mapsto \rho(\xbf)$  to classical adversarial perturbations scaled by the dimension $d_H$ \cite{dowling2024adversarial}. Thus, if the quantum embedding chosen has good scaled adversarial smoothness, it results in a lower value of $\mathcal{S}_{r,p,\epsilon}^{C}$, and thus lower adversarial generalization error.
In contrast, under quantum attacks the choice of quantum embedding only determines the dimension $d_H$, as can be seen from \eqref{eq:quantum_c}. 

\subsubsection{Impact of Quantum Embeddings}
We now  discuss some of the insights gained from the upper bound in Theorem~\ref{thm:upperbound_binary}  about the choice of quantum embedding on adversarial generalization. To this end,  the following 
 proposition upper bounds the scaled excess RC $\mathcal{S}_{r, p,\epsilon}^C$ for the quantum embeddings introduced in Section~\ref{sec:quantumembedding}.
\begin{proposition} \label{prop:quantumembeddings}The scaled excess RC $\mathcal{S}^{C}_{r, p,\epsilon}$ can be upper bounded under amplitude, $L$-layer rotation and dense embeddings respectively as 
\begin{align}
&\mathcal{S}^{C}_{r, p,\epsilon} \nonumber \\& \leq \begin{cases} 2^{1+\lceil\log_2d\rceil}\min\Big\{\frac{\epsilon\max\{1, d^{1/2-1/p}\}}{\min_{x\in\mathcal{X}}\lVert x\rVert_2},1\Big\} \hspace{0.5cm} \mbox{(amplitude)}, \\
2L(2\epsilon)^dd^{-d/p}\hspace{0.4cm} \qquad \qquad \qquad \qquad \mbox{(L-layer angle)},\\
2L(2\sqrt{2}\epsilon)^{d/2}\max\{d^{-d/4},d^{-d/2p}\} \hspace{0.3cm}  \mbox{(L-layer dense)}.\end{cases} \nonumber\end{align}
\end{proposition} 

The proposition above derives the behavior of the scaled excess RC in \eqref{eq:adversarial_binary_upperbound} under various quantum embeddings as detailed in  Table~\ref{tab:comparison}. This in turn can be used to study the impact of the quantum embedding on the ARC $\Rcal(\Fcal_{r,p,\epsilon}^{C})$ and thus on the adversarial generalization error. 

Specifically, under \textbf{classical attacks}, we have the following observations.
In the \textit{high-dimensional} regime, where $d \gg e^{\max\{2,p\}\ln(2\sqrt{2}\epsilon)-1}$, rotation embeddings incur negligible excess RC, requiring a vanishingly small increase in sample complexity to ensure the adversarial generalization to be on par with standard non-adversarial generalization. Indeed, for rotation embeddings, the excess RC \textit{decays at least exponentially with the data dimension $d$}, vanishing in the limit $d\rightarrow\infty$. This implies that, in these regimes, one can leverage the high expressivity of repeated rotation embeddings while incurring a negligible increase in the sample complexity caused by the possible presence of an adversary.

In contrast, for quantum classifiers that employ  \textit{amplitude embedding}, the ARC has an unavoidable dimensional dependence which is at most linear in $d$.

\textit{Remark 2:} The aim of adversarial training is to minimize the adversarial population risk. This is bounded via the ARC with probability $1-\delta$ as 
\begin{align*}
    L^{C}_{r,p,\epsilon}(g)\leq \widehat{L}^{C}_{r,p,\epsilon}(g) +2\Rcal(\ell^{C}_{r,p,\epsilon} \circ \Gcal_r)+3B \sqrt{\frac{1}{2m} \log \frac{2}{\delta}}
\end{align*}
Our analysis can assess which embeddings yield low adversarial Rademacher complexity, but cannot guarantee that they also minimize the \textit{empirical risk}. Indeed, in the non-adversarial setting, it has been shown that minimizing the empirical risk and the generalization error are competing processes with regards to the choice of quantum embedding \cite{banchi2021generalization}. We leave it to future work to assess how to choose an embedding which jointly minimizes the adversarial empirical risk and generalization error.\qed


Finally, when the adversary employs \textbf{quantum attacks}, the excess RC scales at least linearly with the dimension of the Hilbert space. For rotation embeddings, this implies an exponential dependence on the data dimension $d$, whilst amplitude embedding based classifiers incur an additional term with at most quadratic dependence on $d$.

\vspace{-0.3cm}
\subsubsection{Tighter Bounds for Noisy Embeddings}\label{subsec:tight}
In this section, we focus on quantum adversarial attacks and derive tighter bounds on the ARC $\Rcal(\Fcal^Q_{r,p,\epsilon})$, under  the following restricted class of noisy embeddings studied in \cite{georgiou_ISIT24}.
\begin{restatable}{assumption}{asegval}
    \label{assum:1}
  The quantum embedding $\xbf \mapsto \rho(\xbf)$ is such that  the minimum eigenvalue $\lambda_{\min}(\rho(\xbf)) \geq \epsilon$.  
\end{restatable} 
 Assumption~\ref{assum:1} is restrictive as it naturally excludes all pure quantum states.
However, it may reflect well the quantum embeddings produced by NISQ quantum devices.
Under Assumption~\ref{assum:1}, the following theorem gives upper and lower bounds on the ARC $\Rcal(\Fcal^{Q}_{r,p,\epsilon})$. 
\begin{theorem}\label{thm:lowerbound} Consider a $(p,\epsilon)$-quantum adversary for $p \geq 1$ and $\epsilon>0$. Under Assumption~\ref{assum:1}, the  ARC $\Rcal(\Fcal^{Q}_{r,p,\epsilon})$  satisfies the inequalities
\begin{align}\Rcal(\Fcal_r)\leq\Rcal(\Fcal^{Q}_{r,p,\epsilon}) \leq \Rcal(\Fcal_r)+\frac{b\epsilon\max\{1,d_H^{1-1/p-1/r}\}}{\sqrt{m}}. \label{eq:bounds_noisyembedding}
    \end{align}
\end{theorem} 
\begin{proof}
Proof can be found in Appendix~\ref{app:lower}.
\end{proof}
Assumption \ref{assum:1} allows us to evaluate a tighter upper bound on the ARC because in this regime, the adversarial perturbations are universal within each class. In fact, the set of adversarial perturbations corresponding to input $\rho(\xbf)$,  defined as $$P_{\succeq0}=\{\tau:\tau+\rho(\xbf)\succeq0, \rm{Tr}(\tau)=0, \lVert\tau\rVert_{p}\leq\epsilon\}$$  
becomes equivalent to the set $$P=\{\tau: \rm{Tr}(\tau)=0, \lVert\tau\rVert_p\leq\epsilon\},$$ which is independent of $\mathbf{x}$. This allows us to employ H\"older's inequality which together with the supremum over the space of observables $\mathcal{A}_r$ allows for an exact evaluation of the excess RC. On the contrary, when Assumption~ \ref{assum:1} is not satisfied, we need to resort to a covering number analysis which gives an upper bound on the excess RC.

Theorem~\ref{thm:lowerbound} shows that the RC $\Rcal(\Fcal^{Q}_{r,p,\epsilon})$ under quantum attacks is at least as high as that of the non-adversarial function class $\Fcal_r$, for a  class of quantum embeddings satisfying Assumption~\ref{assum:1}. Thus, for these embeddings, uniform convergence in the adversarial setting is provably \textit{at least as hard} as in the non-adversarial setting.
Additionally, note that the upper bound in Theorem~\ref{thm:lowerbound} has a tighter dimensional dependence than the one in \eqref{eq:adversarial_binary_upperbound}, scaling \emph{at most} linearly with $d_H$.  Indeed,  this dimensional dependence can be avoided if $r\geq p/(p-1).$

Finally, we like to remark that finding lower bounds on $\Rcal(\Fcal^{C/Q}_{r,p,\epsilon})$ is in general challenging. For classical neural networks \cite{xiao2022adversarial},  lower bounds on the ARC are conventionally derived by appropriately restricting the function class $\mathcal{F}^C_{r,p,\epsilon}$. However, a similar approach does not extend to quantum classifiers subject to classical attacks, due to the non-linearity in $\bold{x}$ induced by the embedding $\bold{x}\mapsto\rho(\xbf)$.
Additionally, under quantum attacks, the derivation of lower bounds must ensure that the adversarial state is a physical density matrix -- a positive semi-definite, unit-trace, Hermitian matrix.  The positive semi-definite condition $\rho'\succeq0$ in turn requires the evaluation of the eigenvalues of $\rho'$, which is a polynomial equation of order $d_H$. This has no general analytic solutions for $d_H>4$ (more than 2 qubits).
\vspace{-0.2cm}
\section{Sketch of Proof of Theorem~\ref{thm:upperbound_binary}} \label{sec:proof}
We now provide key steps used in the derivation of the upper bound in Theorem~\ref{thm:upperbound_binary} for classical adversarial attacks. The bound for quantum adversarial attacks can be obtained following same steps. We start by noting the following series of relations:
\begin{align} &m\Rcal(\Fcal_{r,p,\epsilon}^{C})\nonumber\\&=  \Ebb_{\boldsymbol{\sigma}}\Biggl[\sup_{f \in \Fcal_r} \sum_{i=1}^m \sigma_i \min_{\xbf': \Vert \xbf_i-\xbf'\Vert_{p} \leq \epsilon}  y_if(\xbf') +y_if(\xbf_i)-y_if(\xbf_i)\Biggr] \nonumber\\& \stackrel{(a)}{\leq}  m\Rcal(\Fcal_r)+ \Ebb_{\boldsymbol{\sigma}}\Biggl[\sup_{f \in \Fcal_r} \sum_{i=1}^m \sigma_i \hspace{-0.2cm}\min_{\xbf': \Vert \xbf_i-\xbf'\Vert_{p} \leq \epsilon}\hspace{-0.2cm} y_i(f(\xbf') -f(\xbf_i))\Biggr]\nonumber \\
&= m\Rcal(\Fcal_r)+ m \Rcal(\tilde{\Fcal}_{r,p,\epsilon}), \label{eq:intermediatestep}\end{align} where $\Rcal(\tilde{\Fcal}_{r,p,\epsilon})$
is the Rademacher complexity of the function class $\tilde{\Fcal}_{r,p,\epsilon}=\{(\xbf,y) \mapsto g(\xbf,y)= \min_{\xbf': \Vert \xbf -\xbf'\Vert_{p} \leq \epsilon} yf(\xbf') -yf(\xbf): f \in {\Fcal_r}\}$ defined as in \eqref{eq:Rcal_Fcal}. The inequality in $(a)$ follows from the subadditivity of the supremum, $\sup(a+b) \leq \sup(a)+\sup(b)$. 

 It now remains to upper bound the term $\Rcal(\tilde{\Fcal}_{r,p,\epsilon})$. To this end, we use a covering number based argument that adopts a two-step approach to find a $\delta$-cover of $\tilde{\Fcal}_{r,p,\epsilon}$.
Critical to this, is the following upper bound on the covering number of the space $\Acal_r$ of Schatten norm-bounded Hermitian observables in \eqref{eq:Ocal}.
\begin{restatable*}{lemma}{lemcover}\label{lem:covering_hermitian_petros}
Consider the set $\Acal_r$ of Hermitian matrices in \eqref{eq:Ocal}. The $\delta$-covering number $\mathcal{N}(\Acal_r, \lVert\cdot\rVert_r,\delta)$ of $\Acal_r$ with respect to the $r$-Schatten norm $\lVert\cdot\rVert_r$ can be upper bounded as
\begin{align}
     \mathcal{N}(\mathcal{A}_r, \lVert\cdot\rVert_r,\delta) \leq \Bigl(3\frac{2^{1/r}b}{\delta} \Bigr)^{d_{H}^2}, \hspace{0.1cm} \mbox{for} \hspace{0.1cm} 0\leq \delta \leq 2^{1/r}b. \nonumber
\end{align}
\end{restatable*} 
\begin{proof}
    Detailed proof can be found in Appendix~\ref{apdx:herm_cover}.
\end{proof}

Equipped with this, we now bound the term $\Rcal(\tilde{\Fcal}_{r,p,\epsilon})$.
 Firstly, obtain a $\delta'=\delta/\mathcal{S}^{C}_{r,p,\epsilon}$-cover of the space $\Acal$ of Hermitian observables with respect to $r$-Schatten norm such that for any $A \in \Acal_r$, there exists $A^c$ in the cover such that $\Vert A-A^c\Vert_{r} \leq \delta'$. To this end, we use Lemma~1. In the second step, for functions $g_A$ and $g_{A^c} \in \tilde{\Fcal}_{r,p,\epsilon}$ defined respectively by observables $A$ and $A^c$, we will show that $|g_A(x,y)-g_{A^c}(x,y)| \leq \delta$. This implies that a $\delta'=\delta/\mathcal{S}^C_{r,p,\epsilon}$-cover of the observable space $\Acal_r$ yields a $\delta$-cover on the adversarial function space $\tilde{\Fcal}_{r,p,\epsilon}$. Using this, we complete the proof by the application of Dudley entropy integral bound \cite{vershynin2018high}.
 Detailed derivation can be found in Appendix~\ref{app:binary_adversary_upperbound}.
 \vspace{-0.2cm}
\section{Extension to Multi-Class Classifiers}
We now extend the previous results to study adversarial generalization of $K$-class quantum classifiers defined in \eqref{eq:multiclass_quantumclassifiers}.  To this end, we follow the standard margin bound framework  used in classical literature \cite{yin2019rademacher, awasthi2020adversarial}. For $K$-class classification problems, let $\mathbf{f}(\xbf)=[f_1(\xbf), \hdots, f_K(\xbf)]$ denote the $K$-length vector-valued function, whose  $k$th entry corresponds to the score $f_k(x)=\Tr(A_k \rho(\xbf))$  assigned to the $k$th class.
The final predicted label is then obtained as $g(\xbf)= \arg \max_{k \in [K]}f_k(x).$
Denote $M(\mathbf{f}(\xbf),y)= f_y(\xbf) - \max_{{k \neq y}}f_k(\xbf)$ as the margin operator that quantifies the difference between the $y$-th class score and the best score assigned to other classes. For a training example $(\xbf,y)$, $\mathbf{f}$ makes correct prediction only if $M(\mathbf{f}(\xbf),y)>0$. We consider the following loss function \cite{yin2019rademacher,awasthi2020adversarial,xiao2022adversarial},
\begin{align}
    \ell(g,\xbf,y)= \phi_{\gamma}(M(\mathbf{f}(\xbf),y)),
\end{align} where $\gamma>0$ and $\phi_{\gamma}: \mathbb{R} \rightarrow [0,1]$ is the $1/\gamma$-Lipschitz, ramp loss defined as
$
    \phi_{\gamma}(t)= 1 \Ibb\{t \leq 0\}+ (1-\frac{t}{\gamma}) \Ibb\{ 0< t <\gamma\},
$ where $\Ibb\{\cdot\}$ denotes the indicator function. 


The following theorem gives an upper bound on the ARC.
\begin{theorem}\label{thm:multiclass_adversarial}
 The ARC of $K$-class margin-based quantum classifier can be upper bounded as
 \begin{align}
     \Rcal(\ell^{C/Q}_{r,p,\epsilon} \circ \Gcal) \leq \frac{2K}{\gamma}\Rcal(\Pi(\Fcal_K)) + \frac{Kb \rob J'(r)}{\gamma\sqrt{m}} 
 \end{align} where $\Rcal(\Pi(\Fcal_K))=\Ebb_{\boldsymbol{\sigma}}\Bigl[ \sup_{k \in [K], f_k \in \Fcal} \frac{1}{m} \sum_{i=1}^m  \sigma_i f_k(\xbf_i)\Bigr]$ can be upper bounded as in Theorem~\ref{thm:binary_nonadversarial} and $J'(r)$ scales as $\mathcal{O}(\sqrt{1+1/r})$.
\end{theorem}
\begin{proof}
Proof can be found in Appendix~\ref{app:multi-class}.
\end{proof}
Theorem~\ref{thm:multiclass_adversarial} extends the main inference from binary to $K$-class classification with an additional multiplicative factor of $K$ on the generalization error bounds.

\section{Numerical Examples}
In this section, we report numerical examples to validate our theoretical findings. Section \ref{sec:embeddings}  evaluates the impact of quantum embeddings on the adversarial generalization error of quantum classifiers under classical attacks, while Section~\ref{sec:impact of noise} investigates the impact of embedding noise on the adversarial generalization error under quantum attacks. Code and data are available at \cite{code_git}.
\subsection{Impact of Embedding on Classical Attacks}\label{sec:embeddings}
\begin{figure}
    \centering
    \includegraphics[width=1.\linewidth]{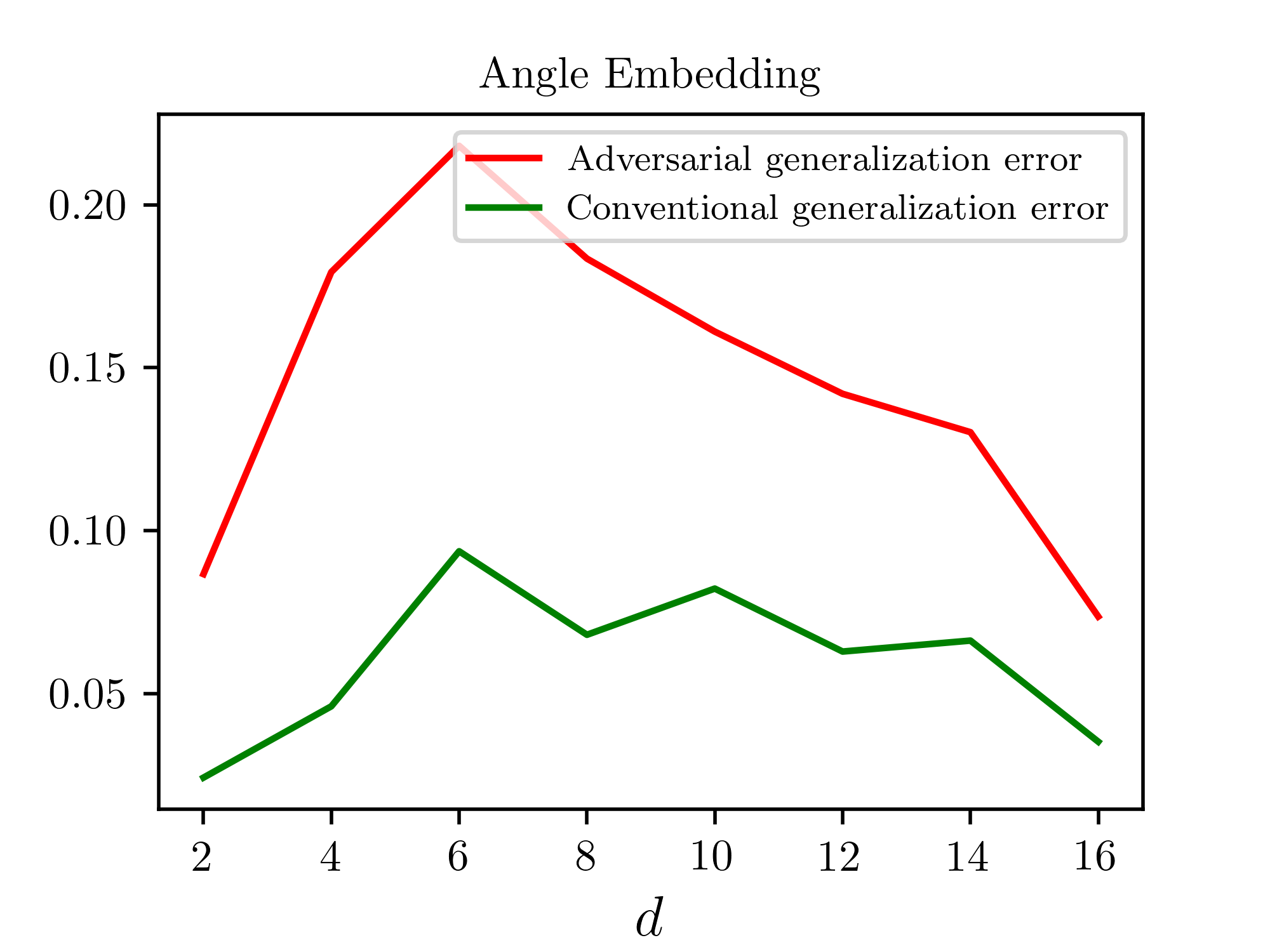}
    \includegraphics[width=1.\linewidth]{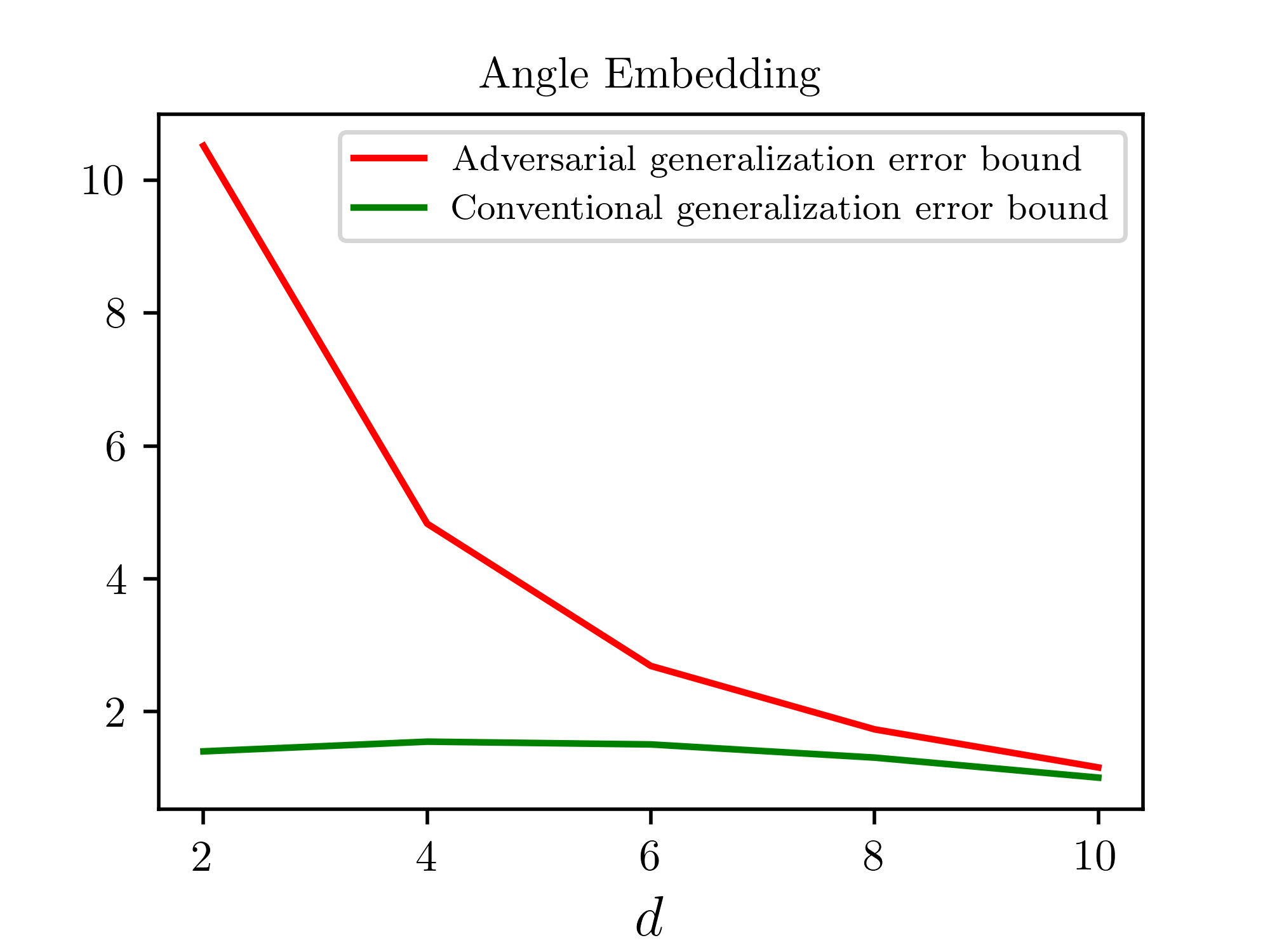}
    \caption{Conventional and adversarial generalization errors (Top) and their respective upper bounds (Bottom) as a function of data dimension $d$. The \emph{angle embedding-based} classifier was adversarially trained on $m=20$ samples against a classical $p=\infty, \epsilon=0.3$ adversary. The conventional and adversarial generalization errors converge in the limit of high dimensional data.}
    \label{fig:angle_gen}
\end{figure}
We consider the task of binary classification with equiprobable class labels $y \in \{0, 1\}$. The input ${\bold{x}}$ is sampled from the multivariate conditional Gaussian distribution $P(\bold{x}|y)=\mathcal{N}(\mu_y, \mathbb{I}_d)$ with mean $\mu_y \in \mathbb{R}^d$ and covariance matrix given by the identity matrix in $d$ dimensions, $\mathbb{I}_d$. For $d=2$, the mean $\mu_y$ of this distribution is
\begin{align*}
    \mu_y= \begin{cases}
        \frac{\pi}{4}(1,1)
        &\text{ if $y = 0$}\\
         \frac{\pi}{4}(1,-1)&\text{ if $y = 1$}
    \end{cases}.
\end{align*}
More generally, for data dimension $d$, the mean $\mu_y$ can be written as
\begin{align*}
    \mu_y = \begin{cases}
        \frac{\pi}{4}I_d&\text{ if $y = 0$}\\
        \frac{\pi}{4}\big(I_{\lfloor d/2\rfloor}\oplus (-I_{d-\lfloor d/2\rfloor})\big)&\text{ if $y = 1$},
    \end{cases}
\end{align*} where $I_n$ denotes a vector of ones in $\mathbb{R}^n$. The above choice ensures that, for even values of $d$, the quantum states obtained via angle and amplitude embedding of the mean of the two classes are orthogonal. Furthermore, this setting allows us to easily vary the data dimension to investigate its impact on the adversarial generalization error.

As the quantum classifier $f(\xbf)$, we consider a class of parametrized observables obtained via the application of a variational quantum circuit,  i.e., a quantum neural network, on the input quantum state followed by fixed measurements. This setting is similar to the one mentioned in Remark 1. Specifically, the quantum neural network has four strongly entangling layers (see PennyLane documentation \cite{pennylane}), followed by a Pauli $Z$ measurement on all qubits, which is an $r=\infty$-bounded observable. We choose as loss function $$
     \ell(g, \rho(\xbf), y) = \frac{1}{1+e^{\alpha f(\bold{x})}},
$$ which is Lipschitz and monotonically increasing, with
 $\alpha=10$.

In Figure~\ref{fig:angle_gen}, we compare the conventional and adversarial generalization errors (left), and the corresponding upper bounds we derived in Theorem~\ref{thm:binary_nonadversarial} and Theorem~\ref{thm:upperbound_binary} (right) as a function of the data dimension $d$. We adversarially train a classifier that employs angle embedding for 40 epochs via the classical Fast Gradient Sign Method (FGSM) attack, which is an approximate $p=\infty$-norm attack. We choose the adversary's strength to be $\epsilon=0.3$, and used a $m=20$-sample training dataset generated according to the multi-variate Gaussian distribution described above.  

Using the adversarially trained classifier, we first evaluate the adversarial generalization error as  the difference between adversarial population risk (approximated using $m=1000$ samples) and the empirical risk. We also evaluate the conventional generalization error of the above trained classifier by evaluating the difference between conventional (non-adversarial) population and empirical risks.   

Figure \ref{fig:angle_gen} shows us that, as predicted by our theoretical bounds in Theorem~\ref{thm:upperbound_binary}, the adversarial and conventional generalization errors for angle embedding-based classifiers converge. This implies that, in the high-dimensional limit, the increase in sample complexity required for adversarially trained classifiers vanishes compared to conventional training. 

In a similar way, for an adversarially trained quantum classifier that uses \emph{amplitude embedding}, Figure \ref{fig:ampl_gen} compares the adversarial and conventional generalization errors with their corresponding upper bounds. Note that for amplitude embedding, the classifier uses $n=\lceil\log_2d\rceil$ qubits. We pad the state vector with 0 entries as a state vector of dimension $d$ is only achievable when $\log_2d$ is an integer.

Figure \ref{fig:ampl_gen} shows us that the adversarial and conventional generalization errors for amplitude embedding classifiers \textit{diverge} as the data dimension $d$ increases. The difference between the adversarial and conventional generalization error  plateaus at higher dimensions.
\begin{figure}
    \centering
    \includegraphics[width=1.\linewidth]{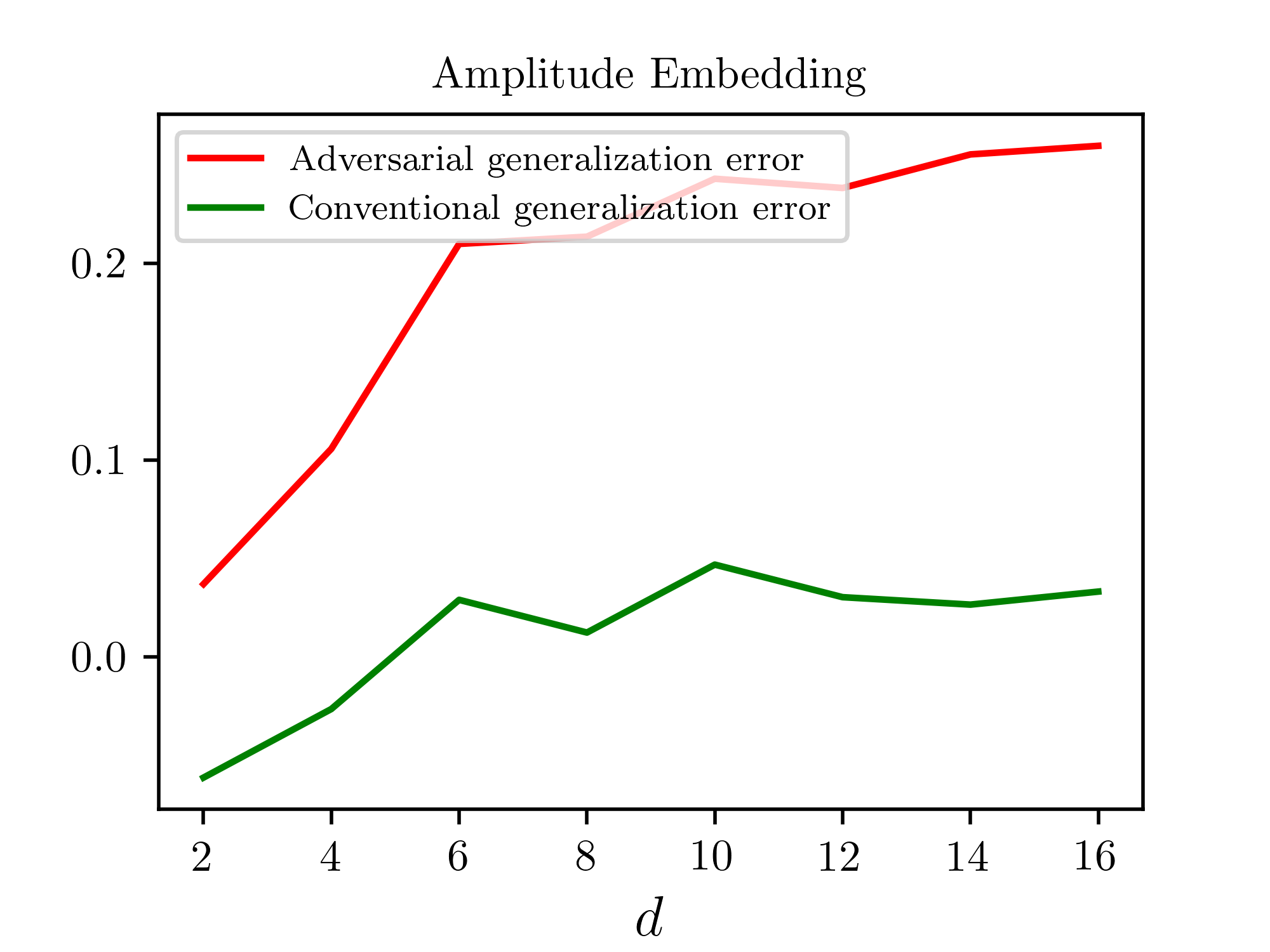}
    \includegraphics[width=1.\linewidth]{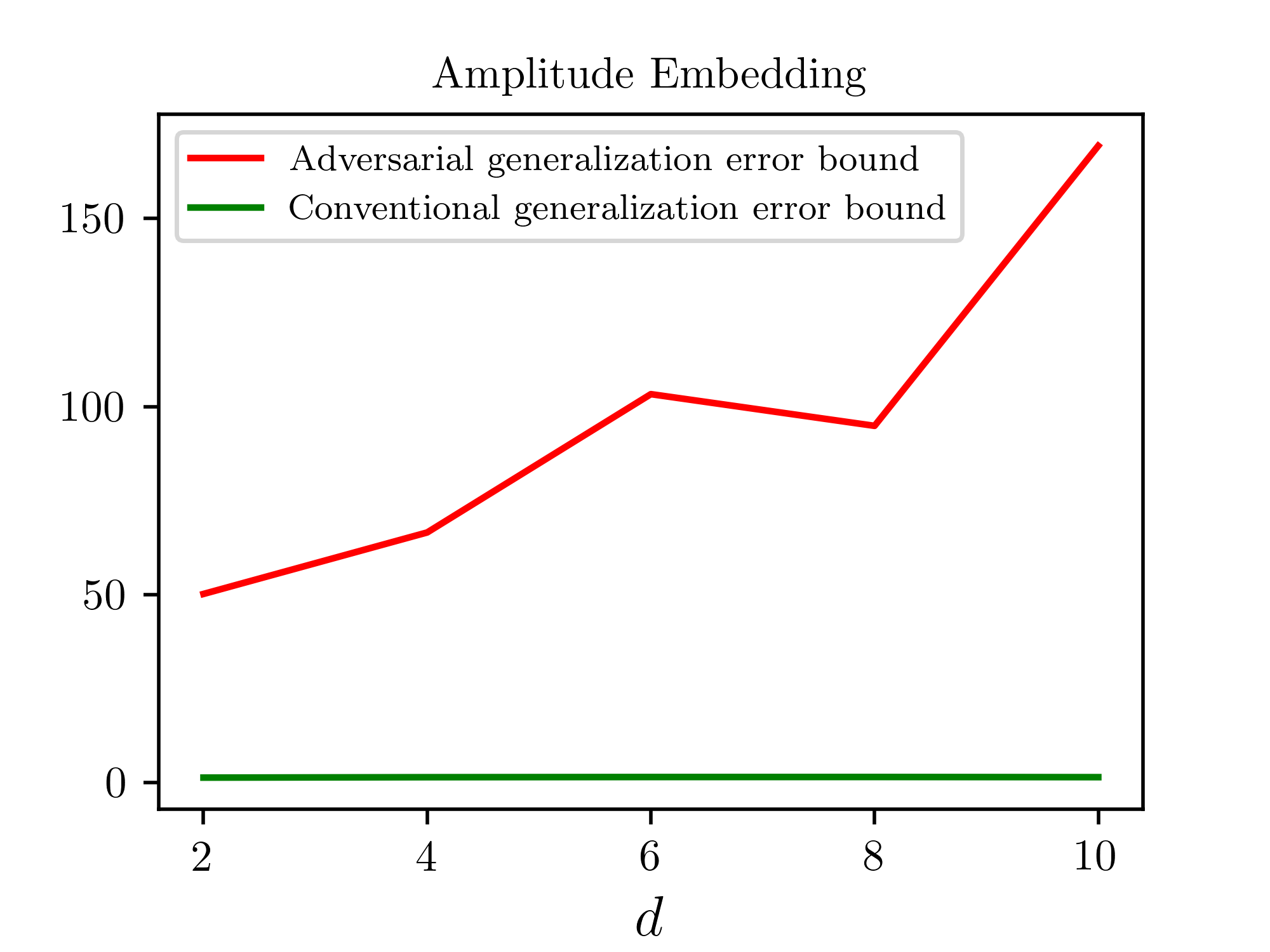}
    \caption{Conventional and adversarial generalization errors (Top) and their respective upper bounds (Bottom) as a function of data dimension $d$. The \emph{amplitude embedding-based} classifier was adversarially trained on $m=20$ samples against a classical $p=\infty, \epsilon=0.3$ adversary. The conventional and adversarial generalization errors diverge in the limit of high dimensional data.}
    \label{fig:ampl_gen}
\end{figure}
This highlights the fact that, whilst uniform convergence bounds are useful to provide a theoretical base for deciding on the choice of embeddings, they are not tight bounds. In order to obtain tighter bounds one may need to use tools outside the uniform convergence framework \cite{understanding_gen}. 
\subsection{Impact of Embedding Noise on Adversarial Generalization Error with Quantum Attacks}\label{sec:impact of noise}
In this section, we investigate the impact of noisy embeddings satisfying Assumption \ref{assum:1}, on the adversarial generalization error.

We again consider the task of binary classification with equiprobable labels $y\in\{0,1\}$ using the same conditional probability distribution $P(\bold{x}|y)$ as in the previous subsection with $d=6$. We use amplitude embedding to map classical input $\xbf$ to a noiseless quantum state $\rho(\xbf)$. This embedding requires $n=3$ qubits, and thus has a Hilbert space dimension $d_H=8$. To prepare noisy quantum states, we apply a global depolarization noise channel on the state $\rho(\xbf)$, obtaining  \begin{align*}
   \rho'(\xbf) &= (1-\lambda_{\min}d_H\rho(\xbf))+\lambda_{\min}\mathbb{I}, \end{align*} 
  where $\lambda_{\min}=0.011$ is the minimum eigenvalue of the density matrix $\rho'(\xbf)$ and $\mathbb{I}$ is the identity matrix. 

We adversarially train a quantum neural network with the same architecture as in Section \ref{sec:embeddings} against an approximate $p=\infty$-Schatten norm quantum adversary of strength $\epsilon=0.001$ with the dataset consisting of quantum states coming from the noiseless embedding. This is implemented via a quantum extension of FGSM that is introduced in  
Appendix~\ref{appdx:QFGSM}. We  evaluate the conventional and adversarial training and test losses  assuming both noiseless and noisy embedding. To do this, we again employ the proposed quantum FGSM adversary with perturbation budget $\epsilon$ varying in the range $\epsilon \in [0.001,0.01]$. We note that this range of $\epsilon$ satisfies Assumption~\ref{assum:1}.

Figure~\ref{fig:noisy} compares the conventional and adversarial generalization errors under noiseless and noisy embedding (top) and the corresponding upper bounds derived in Theorem~\ref{thm:binary_nonadversarial}, Theorem~\ref{thm:upperbound_binary} and Theorem~\ref{thm:lowerbound} (bottom) as a function of the attack strength $\epsilon$. As predicted by our theoretical upper bounds, the adversarial generalization error under noiseless and noisy embeddings  increases with $\epsilon$. Furthermore, it can be seen from Figure~\ref{fig:noisy} (top) that the adversarial (and conventional) generalization error of the noisy embedding is lower than that of the noiseless embedding.  Our bound in Theorem \ref{thm:upperbound_binary} (and Theorem \ref{thm:binary_nonadversarial}) attributes this behavior to the standard RC $\Rcal(\Fcal_r)$, which is known to decrease when the embedding is noisy \cite{banchi2021generalization}. 
\begin{figure}
    \centering
    \includegraphics[width=0.95\linewidth]{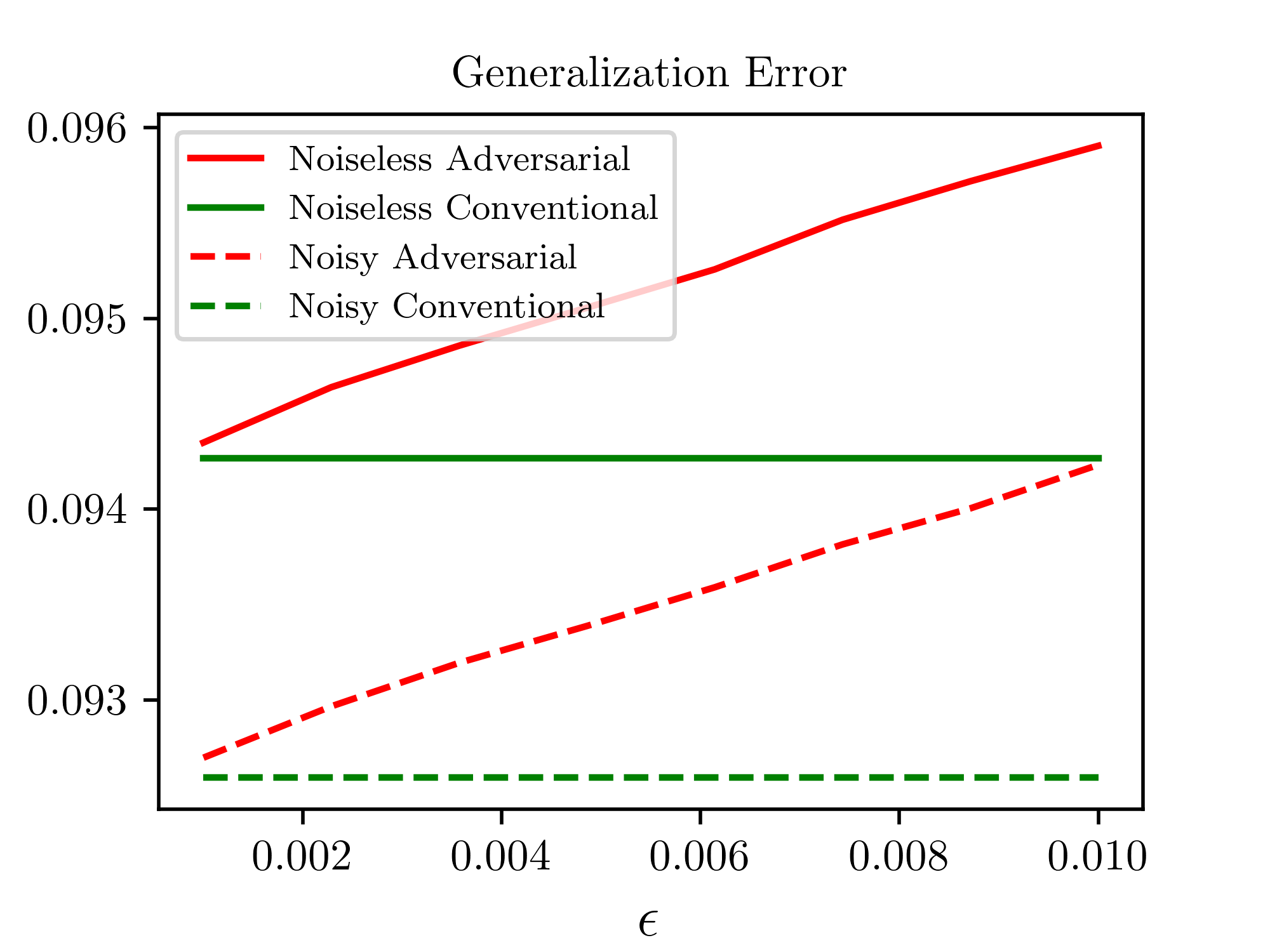}
    \includegraphics[width=0.95\linewidth]{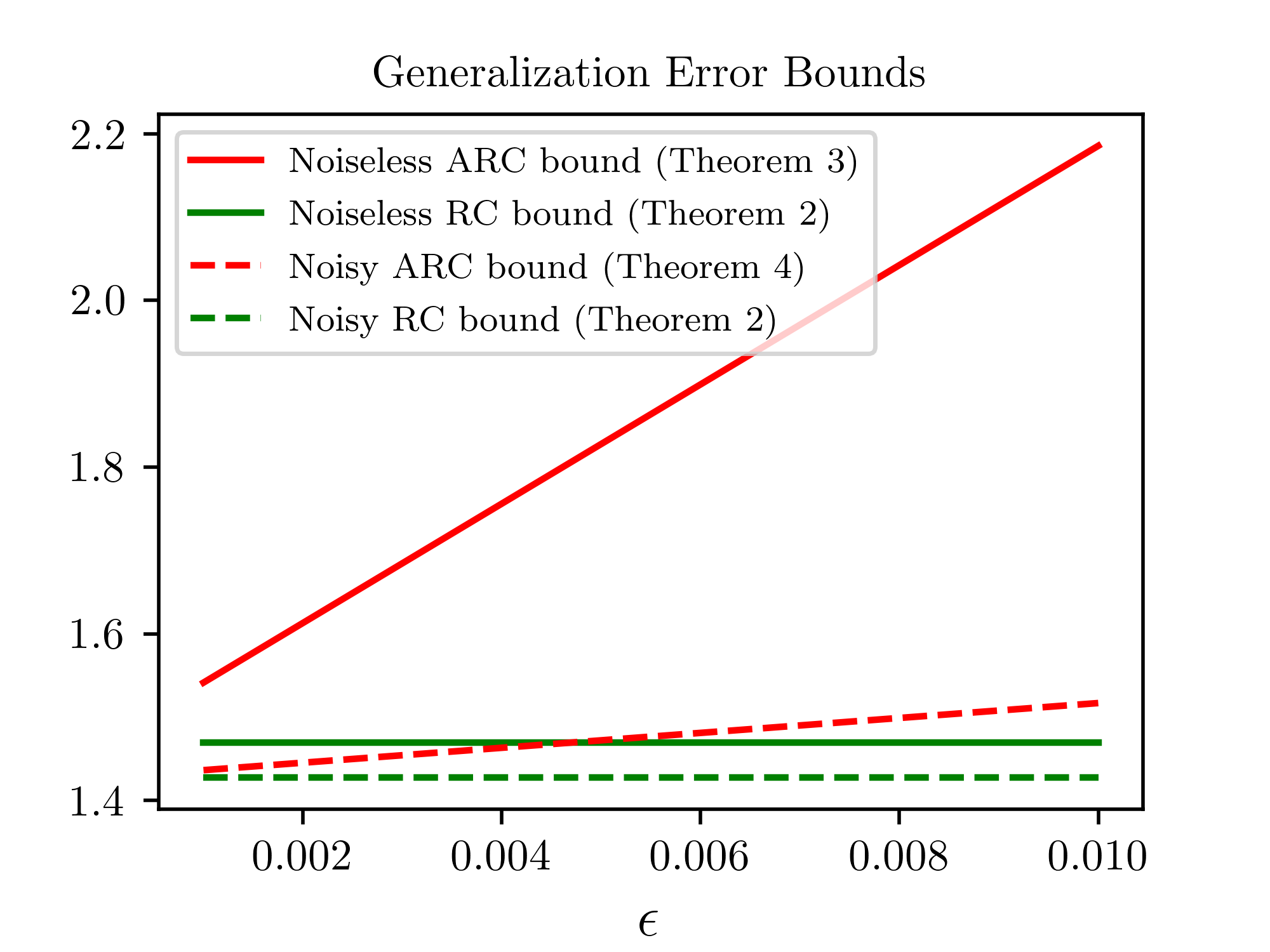}
  
    \caption{
    (Top) Conventional and adversarial  generalization errors for an adversarially trained quantum classifier with a noiseless and noisy embedding and the corresponding bounds (Bottom). The classifier was adversarially trained against the quantum FGSM adversary of strength $\epsilon=0.001$ with $m=20$ samples on $d=6$ dimensional input vectors.}
    \label{fig:noisy}
\end{figure}

\section{Conclusion and Future Works}
This paper has derived novel generalization error bounds for adversarially trained quantum classifiers via an adversarial Rademacher complexity analysis. Our bounds quantify the increase in sample complexity incurred due to adversarial training against classical and quantum adversaries, while also highlighting the impact of the choice of quantum embedding. We investigated how conventional choices of quantum embedding schemes affect adversarial generalization under classical attacks and find that in the limit of high-dimensional input $\xbf$ rotation embedding generalize equally well (asymptotically) to classifiers undergoing standard training.

In the case of quantum attacks we have also provided tighter bounds on the adversarial Rademacher complexity for a restricted class of embeddings, which apply to all $p$-Schatten norm attacks, improving over our previous work \cite{georgiou_ISIT24}. Finally, we derive adversarial generalization error bounds for multi-class classification, which have the same functional form as binary classification bounds, and scale linearly with the number of classes $K$.

Future works include accounting for the effect of quantum gate noise inherent in NISQ devices in adversarial generalization.  Additionally, as noted in \cite{understanding_gen}, moving away from uniform convergence bounds is important for better understanding the generalization capacity of quantum models. Furthermore, analyzing the algorithmic stability of adversarial training can lead insights into the effect of optimization method on adversarial generalization.

\section{Acknowledgments}
The work of OS was supported by an Open Fellowship of the EPSRC with reference EP/W024101/1, by the EPSRC project EP/X011852/1, and by the European Union’s Horizon Europe Project CENTRIC under Grant 101096379.
\appendix
\section{Preliminary Definitions and Results}
In this section, we recall some standard definitions and results which we will use in subsequent sections. This will be used to upper bound the adversarial Rademacher complexity in Theorem~\ref{thm:upperbound_binary}. We start by defining what a covering set is.
\begin{definition}[Covering Set]
    Let $A\subseteq \mathbb{R}^m$. The set $B$ is a $\delta$ covering set of $A$ with respect to $\lVert\cdot\rVert$ (denoted as $\mathcal{C}(A, \lVert\cdot\rVert,\delta)$) if $$\forall a\in A \;\exists\; b\in B : \lVert a-b\rVert\leq\delta$$
\end{definition}
The space of $\delta$-covering sets of a set $A$ defines the $\delta$-covering number for the set, as explained below.
\begin{definition}[Covering Number]
 Let $A\subseteq \mathbb{R}^m$. The $\delta$-covering number  $\Ncal(A, \lVert\cdot\rVert,\delta)$ of $A$ with respect to the norm $\lVert\cdot\rVert$, is the smallest cardinality of a $\delta$-covering set of $A$, i.e.,
\begin{align*}
    \Ncal(A, \lVert\cdot\rVert,\delta) = \underset{B:\;\forall a\in A \;\exists\; b\in B : \lVert a-b\rVert\leq\delta}{\inf} |B|
\end{align*}
    
\end{definition}
The covering number is a useful tool which can be used to upper bound the Rademacher complexity of a set via the Dudley entropy integral bound below.
\begin{theorem}[Dudley Entropy Integral Bound]\label{thm:dudley}
Let $A\subseteq\mathbb{R}^m$ denote a set of $m$-dimensional vectors. Let $B=\underset{a\in A}{\max}\lVert a\rVert_2$. Then, the Rademacher complexity $\Rcal(A) =\Ebb_{\boldsymbol{\sigma}}[\sup_{a \in A} \sum_{i=1}^m \frac{1}{m}\sigma_i a_i]$ can be upper bounded as
\begin{align*}
       \mathcal{R}(A)\leq 12 \int_0^{B/2}\frac{\sqrt{\log(\mathcal{N}(A,\lVert\cdot\rVert_2,\alpha))}}{m}d\alpha
\end{align*}
\end{theorem}
\section{Proof of Lemma~\ref{lem:covering_hermitian_petros}}\label{apdx:herm_cover}
 
 Throughout this section, we use $\Ccal(\Acal, \Vert \cdot \Vert_r,\delta)$ to denote a $\delta$-covering net of the set $\Acal$ with respect to the norm $\Vert \cdot \Vert_r$.

 We first re-state the lemma here and then give the proof.
 \lemcover
\begin{proof}
It is easy to verify that $d_H \times d_H$ Hermitian matrices form a $d_{H}^2$-dimensional real vector space. We can then construct an orthonormal basis $\{E^k\}_{k=1}^{d_{H}^2}$ on this space with respect to the Hilbert-Schmidt inner product $$\langle E^m,E^n\rangle = {\rm Tr}\big((E^m)^\dagger E^n\big)=\delta_{m,n},$$ where $\delta_{m,n}=1$ when $m=n$ and is zero otherwise. A possible choice of basis is $\{E^k\}_{k=1}^{d_H^2} = E_{\text{off-diag}}\cup E_{\text{diag}}$ where
\begin{align}
    (E_{\text{off-diag}})_{i,j} &= \Biggl\{\frac{\delta_{i,m}\delta_{j,n}+\delta_{i,n}\delta_{j,m}}{\sqrt{2}}, \nonumber \\& \frac{i\delta_{i,m}\delta_{j,n}-i\delta_{i,n}\delta_{j,m}}{\sqrt{2}}\Biggr\}_{n=1, m\neq n}^{d_\mathcal{H}}\label{eq:offdiag}
\end{align}
and
\begin{align}
(E_{\text{diag}})_{i,j}&=\Biggl\{\frac{\delta_{i,2n-1}\delta_{j,2n-1}+\delta_{i,2n}\delta_{j,2n}}{\sqrt{2}}, \nonumber \\&\frac{\delta_{i,2n-1}\delta_{j,2n-1}-\delta_{i,2n}\delta_{j,2n}}{\sqrt{2}}\Biggr\}_{n=1}^{d_\mathcal{H}/2}\label{eq:diag}
\end{align}
with the subscript $i,j$ indicating the $i,j$-th element of the basis matrices, and the iteration over $m,n$ generates the individual basis matrices. Visually the set in \eqref{eq:offdiag} consists of off-diagonal matrices of the form
\begin{align*}
    \begin{pmatrix}
0 &. &. &. &0\\
.&.& &\frac{1}{\sqrt{2}} &.\\
. & &. & &.\\
.&\frac{1}{\sqrt{2}} & & .&.\\
0 & .& .& .&0\\
\end{pmatrix},\begin{pmatrix}
0&. &. & .&0\\
.& .& &\frac{i}{\sqrt{2}} &.\\
.& &. & &.\\
.& \frac{-i}{\sqrt{2}}& & .&.\\
0& .& .& .&0\\
\end{pmatrix},
\end{align*}
while the set in \eqref{eq:diag} consists of diagonal matrices of the form
\begin{align*}
\begin{pmatrix}
0&. &. &. &0\\
.& .& & &.\\
.& &\frac{1}{\sqrt{2}} & &.\\
.& & &\frac{1}{\sqrt{2}} &.\\
0&. &. &. &0\\
\end{pmatrix},\begin{pmatrix}
0&. &. &. &0\\
.& .& & &.\\
.& &\frac{1}{\sqrt{2}} & &.\\
.& & & \frac{-1}{\sqrt{2}}&.\\
0& .& .& .&0\\
\end{pmatrix}
\end{align*}
which together form a complete orthonormal basis

Using the above orthonormal basis, we now proceed to obtain a $\delta$-cover of the space $\Acal_r$ of Hermitian matrices. To this end, we first note that any $U\in\mathcal{A}$ can be written as a linear combination  $\sum_{k=1}^{d^2_H} u_k E^k$ of the basis matrices with $u=(u_1,\hdots, u_{d^2_H})$ denoting the real-valued vector. We then employ a well-known result from \citep{vershynin2018high} which bounds the covering number of a real vector space: The $\delta$-covering number of a the set $B_R^D(\bold{x}) = \{\bold{x}':\bold{x}'\in\mathbb{R}^D,\lVert \bold{x}'-\bold{x}\rVert_1\leq R\}$ can be upper bounded as 
\begin{align}
    \mathcal{N}(B^D_R(\xbf), \Vert \cdot \Vert_1,\delta) &\leq \Bigl(1+ \frac{2R}{\delta} \Bigr)^D \leq \Bigl(\frac{3R}{\delta} \Bigr)^D, \label{eq:vec_cov}
\end{align}for $0 < \delta \leq R$.

In order to use this bound, we need to determine the radius $R$ of the vector space of interest. To this end, we first note that the eigenvalues $\lambda_A$ of matrices $A\in\mathcal{A}_r$ have bounded norm, i.e., $|\lambda_{A\in\mathcal{A}_r}|\leq\lVert A\rVert_\infty\leq\lVert A\rVert_r\leq b$.  
Now, consider the following set of matrices \begin{align}\mathcal{Q}=\Big\{Q: -b&\leq Q_{ii}-\sum_{j\neq i}|Q_{ij}|\leq Q_{ii}+\sum_{j\neq i}|Q_{ij}|\leq b \;\forall i \nonumber \\& Q=Q^\dagger\Big\}.
\end{align} Then, from 
Gershgorin's circle theorem \cite{horn13}, we have that for any eigenvalue $\lambda_Q$  of $Q$, there exists $i$ such that the following inequalities hold:
\begin{align*}
     Q_{ii}-\sum_{j\neq i}|Q_{ij}|\leq\lambda_Q\leq Q_{ii}+\sum_{j\neq i}|Q_{ij}|.
\end{align*}
This implies that for $Q\in \mathcal{Q}$, we have $-b\leq\lambda_{Q}\leq b$ which is a sufficient condition to ensure $\lVert Q\rVert_\infty\leq b$.  Therefore $\mathcal{A}_r= \{A:A=A^\dagger, \lVert A\rVert_r \leq b\}\subseteq\{A:A=A^\dagger, \lVert A\rVert_\infty \leq b\}\subseteq\mathcal{Q}$,  
 which in turn yields that $\mathcal{N}(\mathcal{O},\lVert\cdot\rVert_r,\delta)\leq\mathcal{N}(\mathcal{Q},\lVert\cdot\rVert_r,\delta)$. We therefore proceed to upper bound the covering number of the set $\mathcal{Q}$.
 
To this end, we prove that the diagonal elements of matrices $Q$ are bound as $|Q_{ii}|\leq b$ using the eigen decomposition of $Q$ as $Q = UDU^\dagger$. This implies that 
 \begin{align*} Q_{ii}&=\sum_{j}U_{ij}D_{jj}(U^\dagger)_{ij}=\sum_{j}U_{ij}U_{ij}^{*}D_{jj}=\sum_{j}|U_{ij}|^2D_{jj}.
 \end{align*}
 By construction, we have $-b\leq D_{jj}\leq b$, which implies that  $ |Q_{ii}|\leq b$. Additionally, 
 Gershgorin's circle theorem gives that for $i\neq j$,
\begin{align*}
    |Q_{ij}|\leq \sum_{j\neq i}|Q_{ij}|\leq b-Q_{ii}
    \leq 2b.
\end{align*} 
Combining these results with the basis we have constructed, gives that the radius $R$ of our $d_H^2$-dimensional real vector space $\mathcal{Q}$ is $R=2\sqrt{2}b$. All that remains now is to relate the covering number of the space of matrices $\mathcal{Q}$ to the covering number of the underlying vector space. 

To this end, assume that there exists a $2^{1/2-1/r}\delta$-cover $\Ccal(B_R^{d^2_H}(0), \lVert\cdot\rVert_1, 2^{1/2-1/r}\delta)$ of the vector space with respect to the $\ell_1$ distance. Now pick a matrix $U$ from the space $\mathcal{Q}$. This can be written as $U =\sum_{k=1}^{d^2_H}u_k E^k$. We prove that the assumed cover of the vector space implies a $\delta$-cover of the space $\mathcal{Q}$ with respect to the $\lVert\cdot\rVert_r$ distance by showing that there exists a matrix $V$ generated from $\Ccal(B_R^{d^2_H}(0), \lVert\cdot\rVert_1, 2^{1/2-1/r}\delta)$ which is $\delta$-close to $U$. To this end, we have the following series of relations:
\begin{align*}
    \lVert U-V\rVert_r &= \lVert\sum_{k=1}^{d_{H}^2}(u_k-v_k)E^k\rVert_r \\&\leq\sum_{k=1}^{d_H^2}|u_k-v_k|\lVert E^k\rVert_r \\
    & \stackrel{(a)}{\leq} \sum_{k=1}^{d_{H}^2}|u_k-v_k| 2^{1/r-1/2} \\
    & \stackrel{(b)}{=} \Vert u-v\Vert_1 2^{1/r-1/2}  \stackrel{(c)}{\leq} \delta
\end{align*} where the inequality $(a)$ follows since $\lVert E^k\rVert_q = 2^{1/q-1/2}$ by construction, $(b)$ is by the definition of $l_1$-distance between two vectors $u=(u_1,\hdots, u_{d^2_H})$ and $v=(v_1,\hdots,v_{d_H^2})$ and $(c)$ follows from the existence of $\Ccal(B_R^{d^2_H}(0), \lVert\cdot\rVert_1, 2^{1/2-1/r}\delta)$. We have thus proved that  a $2^{1/2-1/r}\delta$-cover of $B_R^{d^2_H}(0)$ yields a $\delta$-cover of $\mathcal{Q}$. 
Hence, the $\delta$-covering number $\mathcal{N}(\mathcal{Q}, \lVert\cdot\rVert_r, \delta)$ of the set $\mathcal{Q}$ with respect to the $r$-Schatten norm  is upper bounded by the $2^{1/2-1/r}\delta$-covering number $\mathcal{N}(B^{d_{H}^2}_{2\sqrt{2}b}(0), \lVert\cdot\rVert, 2^{1/2-1/r}\delta)$ of the vector space $\mathbb{R}^{d_H^2}$ with radius $2\sqrt{2}b$ with respect to the $\ell_1$ norm. Putting everything together we get
\begin{align*}
    \mathcal{N}(\mathcal{A}, \lVert\cdot\rVert_r, \delta)&\leq\mathcal{N}(\mathcal{Q}, \lVert\cdot\rVert_r, \delta)\\&\leq \mathcal{N}(B^{d_{H}^2}_{2\sqrt{2}b}(0), \lVert\cdot\rVert, 2^{1/2-1/r}\delta)\\&\leq \Bigl(1+\frac{2^{2+1/r}b}{\delta}\Bigl)^{d_{H}^2}\leq\Bigl(3\frac{2^{1+1/r}b}{\delta}\Bigl)^{d_{H}^2},
\end{align*}
for $0\leq \delta \leq 2^{1+1/r}b.$
\end{proof}
\section{Proof of Theorem~\ref{thm:binary_nonadversarial}}\label{app:binary_nonadversarial}
We upper bound the RC $\Rcal(\Fcal_r)$ of the function class $\mathcal{F}_r$ as follows:
\begin{align}
    \Rcal(\Fcal_r)&= \Ebb_{\boldsymbol{\sigma}}\Bigl[\sup_{f \in \mathcal{F}_r}\frac{1}{m} \sum_{i=1}^m \sigma_i y_i f(\xbf_i) \Bigr] \nonumber \\
    &\stackrel{(a)}{=}\Ebb_{\boldsymbol{\sigma}}\Bigl[\sup_{A\in \Acal_r}\frac{1}{m} \sum_{i=1}^m \sigma_i\Tr(A \rho(\xbf_i)) \Bigr] \nonumber \\
    & = \Ebb_{\boldsymbol{\sigma}}\Bigl[\sup_{A\in \Acal_r} \Tr \Bigl(A (\frac{1}{m} \sum_{i=1}^m \sigma_i \rho(\xbf_i)) \Bigr) \Bigr] \nonumber \\
    & \stackrel{(b)}{\leq}\Ebb_{\boldsymbol{\sigma}}\Bigl[ \sup_{A\in \Acal_r} \Vert A \Vert_r  \biggl \Vert \frac{1}{m} \sum_{i=1}^m \sigma_i \rho(\xbf_i) \biggr \Vert_{\frac{r}{r-1}} \Bigr] \nonumber \\
    &= \frac{b}{m} \Ebb_{\boldsymbol{\sigma}}  \biggl \Vert \sum_{i=1}^m \sigma_i \rho(\xbf_i) \biggr \Vert_{\frac{r}{r-1}}  \nonumber\\ \nonumber\\
    &\stackrel{(c)}{\leq}\frac{b}{m}\begin{cases}
        \sqrt{2\ln(d_H)\lVert \sum_{i=1}^m\rho(\bold{x}_i)^2\rVert_\infty}\;&r=1\\ \\
        B_{r/(r-1)}\lVert \sqrt{\sum_{i=1}^m\rho(\bold{x}_i)^2}\rVert_{\frac{r}{r-1}},\;&2\leq r<\infty\\ \\\lVert\sqrt{\sum_{i=1}^m\rho(\bold{x}_i)^2}\rVert_1,\;&r=\infty
    \end{cases} \nonumber,
\end{align} where $(a)$ follows since $\sigma_i y_i =\sigma_i$ holds due to the symmetric property and $(b)$ from matrix H\"older's inequality. Inequality $(c)$ follows from Jensen's inequality and from the operator Khintchine's inequalities \citep{banchi2023statistical} for $2\leq r\leq\infty$ as
\begin{align*}
&\Ebb_{\boldsymbol{\sigma}}  \biggl \Vert \sum_{i=1}^m \sigma_i \rho(\xbf_i) \biggr \Vert_{\frac{r}{r-1}} \stackrel{\text{Jensen's}}{\leq} \Bigg[\Ebb_{\boldsymbol{\sigma}}  \biggl \Vert \sum_{i=1}^m \sigma_i \rho(\xbf_i) \biggr \Vert_{\frac{r}{r-1}}^{\frac{r}{r-1}}\Bigg]^{\frac{r-1}{r}}\\ &\stackrel{\text{Khintchine's}}{\leq}\begin{cases}
        B_{\frac{r}{r-1}}\lVert \sqrt{\sum_{i=1}^m\rho(\bold{x}_i)^2}\rVert_{\frac{r}{r-1}},\;&2\leq r<\infty\\\\\lVert\sqrt{\sum_{i=1}^m\rho(\bold{x}_i)^2}\rVert_1,\;&r=\infty,
    \end{cases}
\end{align*} 
where we have defined $B_{\frac{r}{r-1}}=2^{-1/4}\sqrt{\frac{\pi r}{e(r-1)}}$. For $r=1$,   Tropp's inequality \citep{banchi2023statistical} gives that \begin{align*}
    \Ebb_{\boldsymbol{\sigma}}\Big \lVert \sum_{i=1}^m\sigma_i\rho(\bold{x}_i)\Big\rVert_\infty\leq\sqrt{2\ln(d_H)\Big\lVert\sum_{i=1}^m\rho(\bold{x}_i)^2\Big\rVert_\infty}.
\end{align*}

\section{Proof of Theorem~\ref{thm:upperbound_binary}} \label{app:binary_adversary_upperbound}
In this section, we build on the proof sketch in Section~\ref{sec:proof} and first derive an upper bound on the Rademacher complexity
\begin{align}
\Rcal(\tilde{\Fcal}_{r,p,\epsilon})
     & = \frac{1}{m}\Ebb_{\boldsymbol{\sigma}}\Biggl[\sup_{g \in \tilde{\Fcal}_{r,p,\epsilon}} \sum_{i=1}^m \sigma_i h(\xbf_i,y_i) \Biggr]
\end{align} of the excess adversarial loss function space $\tilde{\Fcal}_{r,p,\epsilon}$ under  $(p,\epsilon)$-classical adversarial attack. Recall that \begin{align*}\tilde{\Fcal}_{r,p,\epsilon}=\{&(\xbf,y) \mapsto h(\xbf,y)=\\&\min_{\xbf': \Vert \xbf -\xbf'\Vert_{p} \leq \epsilon} yf(\xbf') -yf(\xbf): f \in {\Fcal_r}\}
\end{align*}
with ${\Fcal_r}=\{\xbf \mapsto f(\xbf)=\Tr(A\rho(\xbf)): A \in \Hcal, \Vert A \Vert_r\leq b\}$ denoting the function class parameterized by Hermitian observables $A$. Note that similar steps can be followed to get the upper bound under quantum adversarial attacks. 

To obtain an upper bound on the Rademacher complexity, we proceed via the following steps.

\emph{\textbf{Step 1} Obtain a $\delta/\sqrt{m}$-cover of the excess adversarial function space $\tilde{\Fcal}_{r,p,\epsilon}$ using appropriate cover of $\Acal_r$:}
From Lemma~\ref{lem:covering_hermitian_petros}, we get that there exists a $\delta'$-cover $\Ccal(\Acal_r, \Vert \cdot \Vert_r ,\delta')$ of the space $\Acal_r$ of Hermitian observables with cardinality $|\Ccal(\Acal_r, \Vert \cdot \Vert_r ,\delta')| \leq \Bigl( 3 \frac{2^{1/r}b}{\delta}\Bigr)^{d^2_H}$. This implies that for each $U \in \Acal_r$, there exists $V \in \Ccal(\Acal_r, \Vert \cdot \Vert_r ,\delta')$ such that $\Vert U-V\Vert_r \leq \delta'$. 

We now consider the set $\Ccal(\tilde{F}_{r,p,\epsilon},\delta')=\{g(\xbf,y)=\min_{\xbf': \Vert \xbf-\xbf'\Vert_p \leq \epsilon} y(\Tr(V \rho(\xbf))-\Tr(V\rho(\xbf')): V \in \Ccal(\Acal_r, \Vert \cdot \Vert_r ,\delta')\}$.
Let $g \in \tilde{\Fcal}_{r,p,\epsilon}$ be defined via the function $f(\xbf) =\Tr(U\rho(\xbf))$. Then, there exists a $g' \in \Ccal(\tilde{\Fcal}_{r,p,\epsilon},\delta')$ such that
\begin{align*}
    &|g(\xbf,y)- g'(\xbf,y)| \\&=\Bigl| \Bigl( \min_{\xbf': \Vert \xbf-\xbf'\Vert_p \leq \epsilon} y\Tr( U \rho(\xbf')) -y\Tr( U \rho(\xbf)) \Bigr) \\&\hspace{0.4cm}-\Bigl(\min_{\xbf': \Vert \xbf-\xbf'\Vert_p \leq \epsilon} \Tr( V \rho(\xbf'))-y\Tr( V \rho(\xbf)) \Bigr)\Bigr| \\
    & \stackrel{(a)}{=} \Bigl|  y \Tr( U (\rho(\xbf^*_U)-\rho(\xbf))) -  y\Tr( V (\rho(\xbf^*_V)-\rho(\xbf)))\Bigr| \\
    & \stackrel{(b)}{\leq} \Bigl|  y\Tr( U (\rho(\xbf^*_V)-\rho(\xbf)) -  y\Tr( V (\rho(\xbf^*_V)-\rho(\xbf))\Bigr|\\
    & = \Bigl|  y \Tr\Bigl( (U -V)(\rho(\xbf^*_{V})-\rho(\xbf))\Bigr)\Bigr| \\
    & =|y| \Bigl| \Tr\Bigl( (U -V)(\rho(\xbf^*_{V})-\rho(\xbf))\Bigr)\Bigr| \\
    &\stackrel{(c)}{\leq} \lVert U-V \rVert_{r} \lVert \rho(\xbf^*_{V})-\rho(\xbf)\rVert_{\frac{r}{r-1}}\\
    &\stackrel{(d)}{\leq} \lVert U-V \rVert_{r}\mathcal{S}_{r,p,\epsilon}^{Q/C}\\
    & \leq \delta' \mathcal{S}_{r,p,\epsilon}^{Q/C},
\end{align*} 
where in $(a)$ $\rho(\xbf^*_U)$ denotes the minimizer of $\min_{\xbf': \Vert \xbf-\xbf'\Vert_p \leq \epsilon} y \Tr( U \rho(\xbf')) -y \Tr( U \rho(\xbf))$, and similarly for $\rho(\xbf^*_{V})$. In $(b)$ we assume wlog that $y \Tr( U(\rho(\xbf^*_U)-\rho(\xbf))) \leq  y\Tr(U(\rho(\xbf^*_{V})-\rho(\xbf)))$ and in $(c)$ we use H{\"o}lder's inequality. The upper bound in $(d)$ follows from our definition $\mathcal{S}^{C}_{r,p,\epsilon} =\max_{\xbf':\lVert \xbf'-\xbf\rVert_p\leq\epsilon}\lVert \rho(\xbf')-\rho(\xbf)\rVert_{r/(r-1)}$. Finally, the last inequality follows since $V \in \Ccal( \Acal_r, \Vert \cdot \Vert_r ,\delta')$.

The relations above show that with the choice of $\delta'= \frac{\delta}{\sqrt{m} \mathcal{S}_{r,p,\epsilon}^{C}}$, the set $\Ccal(\tilde{\Fcal}_{r,p,\epsilon},\delta')$ is indeed a $\delta/\sqrt{m}$-cover of $\tilde{\Fcal}_{p,\epsilon}$ with respect to $| \cdot |$. In other words, we have $\Ccal(\tilde{\Fcal}_{p,\epsilon}, | \cdot |, \delta/\sqrt{m})=\Ccal(\tilde{\Fcal}_{r,p,\epsilon},\frac{\delta}{\sqrt{m} \mathcal{S}_{r,p,\epsilon}{C}})$.

{\emph{\textbf{Step 2} Obtain a $\delta$-cover of the excess adversarial function space $\tilde{\Fcal}_{r,p,\epsilon}(\Dcal)$ defined with respect to the data set $\Dcal$:} Let 
$$ \tilde{\Fcal}_{r,p,\epsilon}(\Dcal)=\{\mathbf{g}=( g(\xbf_1,y_1), \hdots g(\xbf_m,y_m)): g \in \tilde{\Fcal}_{r,p,\epsilon}\}$$ denote the excess adversarial function space defined with respect to the data set $\Dcal$.  Furthermore, define the set
\begin{align*} \Ccal(\tilde{\Fcal}_{r,p,\epsilon}(\Dcal))=\{&\mathbf{g}=( g(\xbf_1,y_1), \hdots g(\xbf_m,y_m)):\\&g \in \Ccal(\tilde{\Fcal}_{r,p,\epsilon}, |\cdot |,\delta/\sqrt{m})\}\end{align*}
It is then easy to see that for $\mathbf{g} \in \tilde{\Fcal}_{r,p,\epsilon}(\Dcal)$ and $\mathbf{g}' \in \Ccal(\tilde{\Fcal}_{r,p,\epsilon}(\Dcal))$, we get that 
$$\lVert\bold{g}-\bold{g'}\rVert^2=\sqrt{\sum_{i=1}^m \big(g(\xbf_i,y_i)-g'(\xbf_i,y_i)\big)^2}\leq\delta,$$ which implies that the set $\Ccal(\tilde{\Fcal}_{r,p,\epsilon}(\Dcal)) = \Ccal(\tilde{\Fcal}_{r,p,\epsilon}(\Dcal), \Vert \cdot \Vert_2,\delta)$ forms the $\delta$-cover of $\tilde{\Fcal}_{r,p,\epsilon}(\Dcal)$.

Putting the above results together we get that \begin{align*}
    &\mathcal{N}\big(\tilde{\mathcal{F}}_{r.p,\epsilon}(\mathcal{D}),\lVert\cdot\rVert_2,\delta\big)\leq\mathcal{N}\big(\tilde{\mathcal{F}}_{r,p,\epsilon},\lVert\cdot\rVert_1,\delta/\sqrt{m}\big)\\&\leq\mathcal{N}\big(\mathcal{O}, \lVert\cdot\rVert_r, \delta/(\mathcal{S}_{r,p,\epsilon}^{C}\sqrt{m})\big) \leq \Biggl(3\frac{2^{1/r}\sqrt{m}\mathcal{S}_{r,p,\epsilon}^{C}b}{\delta}\Biggl)^{d_{H}^2}, \label{eq:coveringnumber_dudley}
\end{align*}where the last inequality follows from Lemma~\ref{lem:covering_hermitian_petros}.

\emph{\textbf{Step 3} Application of Dudley Entropy Integral Bound:}
We now employ Dudley entropy integral bound in Theorem~\ref{thm:dudley}, to get that
\begin{align}
&\Rcal(\tilde{\Fcal}_{r,p,\epsilon})
      = \frac{1}{m}\Ebb_{\boldsymbol{\sigma}}\Biggl[\sup_{\mathbf{g} \in \tilde{\Fcal}_{r,p,\epsilon}(\Dcal)} \boldsymbol{\sigma}^{\top} \mathbf{g} \Biggr] \nonumber\\&\leq 12 \int_0^{D/2} \hspace{-0.2cm}\frac{\sqrt{\log(\mathcal{N}(\Tilde{\Fcal}_{r,p,\epsilon}(\Dcal),\lVert\cdot\rVert_2,\alpha))}}{m}d\alpha,
\end{align} where
\begin{align*}
   D&=\max_{g \in \tilde{\Fcal}_{r,p,\epsilon}} \sqrt{\sum_{i=1}^mg(\xbf_i,y_i)^2}\\
    &=\max_{A \in \Acal_r} \sqrt{\sum_{i=1}^m\big({\rm Tr}A(\rho(\xbf^*_i)-\rho(\xbf_i)\big)^2}\\
    &\leq\max_{A \in \Acal_r} \sqrt{\sum_{i=1}^m\big(\lVert A\rVert_r\lVert(\rho(\xbf^*_i)-\rho(\xbf_i)\rVert_{\frac{r}{r-1}}\big)^2}\\
    &\leq\sqrt{\sum_{i=1}^m(\mathcal{S}_{r,p,\epsilon}^C)^2}\leq b\sqrt{m}\mathcal{S}_{r,p,\epsilon}^C.
\end{align*}
Now, substituting the covering number from Lemma~\ref{lem:covering_hermitian_petros} and evaluating the integral in \eqref{eq:coveringnumber_dudley}, we get the following upper bound,
\begin{align*}
    \mathcal{R}(\tilde{\mathcal{F}}_{r,p,\epsilon})&\leq12\int_0^{\sqrt{m}\mathcal{S}_{r,p,\epsilon}^Cb/2}\frac{\sqrt{d_{H}^2\log\Big(3\frac{2^{1/r}\sqrt{m}\mathcal{S}_{r,p,\epsilon}^Cb}{\alpha}\Big)}}{m}d\alpha\\
    &=36\frac{2^{1/r}\mathcal{S}_{r,p,\epsilon}^Cbd_{H}}{\sqrt{m}}\int_0^{2^{-1-1/r}/3}\sqrt{\log\Big(\frac{1}{\alpha}\Big)}d\alpha\\
    &=36\frac{2^{1/r}\mathcal{S}_{r,p,\epsilon}^Cbd_{H}}{\sqrt{m}}\Biggl(\frac{2^{-1/r}}{6}\sqrt{\log(2^{1/r}6)}\\
&\hspace{0.4cm}+\frac{\sqrt{\pi}}{2}\Big(1-{\rm erf}\big(\sqrt{\log(2^{1/r}6)}\big)\Big)\Biggl), 
\end{align*} where ${\rm erf}(z)=\frac{2}{\sqrt{\pi}}\int_0^ze^{-t^2}dt.$ Using this in \eqref{eq:intermediatestep} then yields the required upper bound. 

\section{Proof of Proposition~\ref{prop:quantumembeddings}}\label{app:prop}
In this section, we upper bound the the quantity $\mathcal{S}^{C}_{r,p,\epsilon}=\max_{\xbf':\lVert\rho(\xbf)-\rho(\xbf')\rVert_p\leq\epsilon}\lVert\rho(\xbf)-\rho(x')\rVert_{r/(r-1)}$, for specific types of embeddings. To this end, we note that for pure state embeddings of the form $\rho(\xbf)=\vert \psi(\xbf) \rangle \langle \psi(\xbf)\vert$, this term evaluates as
\begin{align*}
\mathcal{S}^{C}_{r,p,\epsilon}\leq d_H\lVert\rho(\xbf)-\rho(\xbf')\rVert_1=2d_H\sqrt{1-|\big\langle\psi(\xbf+\delta \xbf)|\psi(\xbf)\rangle\big|^2}. 
\end{align*}
Using the above, we now upper bound $\mathcal{S}^{C}_{r,p,\epsilon}$ for specific quantum embeddings.
\subsection{Amplitude embedding}
Recall that amplitude embedding describes the mapping 
\begin{align*}
    \xbf \mapsto |\psi(\xbf)\rangle=\sum_{n=1}^d\frac{\xbf_n}{\lVert \xbf\rVert_2}|n\rangle.
\end{align*}
A direct computation yields that
\begin{align*}
    |\big\langle\psi(\xbf+\delta \xbf)|\psi(\xbf)\rangle\big|^2=\sum_{i,j=1}^d\frac{\xbf_i(\xbf_i+\delta \xbf_i)\xbf_j(\xbf_j+\delta \xbf_j)}{\lVert \xbf\rVert_2^2\;\lVert \xbf+\delta \xbf\rVert_2^2}
\end{align*}
where we write $\xbf'=\xbf+\delta \xbf$. Thus
\begin{align*}
    &\lVert\rho(\xbf)-\rho(\xbf')\rVert_1 \\&=2\sqrt{1-\sum_{i,j=1}^d\frac{\xbf_i(\xbf_i+\delta \xbf_i)\xbf_j(\xbf_j+\delta \xbf_j)}{\lVert \xbf\rVert_2^2\;\lVert \xbf+\delta \xbf\rVert_2^2}}\\
&=2\sqrt{\sum_{i,j=1}^d\frac{\xbf_i^2(\xbf_j+\delta \xbf_j)^2-\xbf_i(\xbf_i+\delta \xbf_i)\xbf_j(\xbf_j+\delta \xbf_j)}{\lVert \xbf\rVert_2^2\;\lVert \xbf+\delta \xbf\rVert_2^2}} \\
    &=2\sqrt{\frac{\lVert \xbf\rVert_2^2\;\lVert \xbf+\delta \xbf\rVert_2^2-|\langle \xbf,\xbf+\delta \xbf\rangle|^2}{\lVert \xbf\rVert_2^2\;\lVert \xbf+\delta \xbf\rVert_2^2}} \end{align*} \begin{align*}
    &=2\sqrt{\frac{\lVert \xbf\rVert_2^2\;\lVert\delta \xbf\rVert_2^2-|\langle \xbf,\delta \xbf\rangle|^2}{\lVert \xbf\rVert_2^2\;\lVert \xbf+\delta \xbf\rVert_2^2}}\\
    &=2\sqrt{\frac{\lVert\delta \xbf\rVert_2^2(1-\cos^2(\theta)}{\lVert \xbf\rVert_2^2+\lVert \delta \xbf\rVert_2^2+2\lVert \xbf\rVert_2\lVert\delta \xbf\rVert_2\cos(\theta)}}\\&\stackrel{(a)}{\leq}2\min\Big\{\frac{\lVert \delta \xbf\rVert_2}{\lVert \xbf\rVert_2},1\Big\}\\&\stackrel{(b)}{\leq}2\min\Big\{\frac{\epsilon\max\{1, d^{1/2-1/p}\}}{\min_{\xbf\in\mathcal{X}}\lVert \xbf\rVert_2},1\Big\},
\end{align*} where
inequality (a) follows from minimizing over the angle $\theta$ between the vectors $\xbf$ and $\delta \xbf$, while (b) follows from maximizing over $\delta \xbf$, and $\xbf\in\mathcal{X}$. Therefore
\begin{align}
    \mathcal{S}_{r,p,\epsilon}^C\leq2\min\Big\{\frac{\epsilon\max\{d, d^{3/2-1/p}\}}{\min_{\xbf\in\mathcal{X}}\lVert \xbf\rVert_2},d\Big\}.
\end{align}

\subsection{Data re-uploading with angle and dense embeddings}

Consider an $L$-layer single qubit re-uploading embedding circuit $\xbf\mapsto\rho(\bold{\xbf})=|\phi(\bold{\xbf})\rangle\langle\phi(\bold{\xbf})|$ where
\begin{align*}
    |\phi(\bold{\xbf})\rangle=\prod_{l=1}^L\big[V_l U(\bold{\xbf})\big]|\bold{0}\rangle,\;|\bold{0}\rangle\coloneqq\bigotimes_{i=1}^d|0\rangle
\end{align*}
We proceed with bounding $\lVert\rho(\bold{x})-\rho(\bold{x}')\rVert_1$ which we will do by induction. Assume the solution is $\lVert\rho(\bold{x})-\rho(\bold{x}')\rVert_1\leq L\lVert \mathcal{U}(\bold{x})-\mathcal{U}(\bold{x}')\rVert_\diamond$
\begin{proof}We begin by defining the diamond distance $$\lVert \mathcal{U}(\bold{x})-\mathcal{U}(\bold{x}')\rVert_\diamond\coloneqq\sup_\rho\lVert U^\dagger(\bold{x})\rho U(\bold{x})-U^\dagger(\bold{x}')\rho U(\bold{x}')\rVert_1$$
    \textit{Step $1$:} Check for $L=1$.
    \begin{align*}
        \lVert\rho(\bold{x})-\rho(\bold{x}')\rVert_1&=\lVert VU(\bold{x})|\bold{0}\rangle\langle\bold{0}|U^\dagger(\bold{x})V^\dagger\\&\hspace{0.4cm}-VU(\bold{x}')|\bold{0}\rangle\langle\bold{0}|U^\dagger(\bold{x}')V^\dagger\rVert_1\\&\stackrel{(a)}{=}\lVert U(\bold{x})|\bold{0}\rangle\langle\bold{0}|U^\dagger(\bold{x})-U(\bold{x}')|\bold{0}\rangle\langle\bold{0}|U^\dagger(\bold{x}')\rVert_1\\&\leq\lVert \mathcal{U}(\bold{x})-\mathcal{U}(\bold{x}')\rVert_\diamond
    \end{align*}
    where in $(a)$ we use the unitary invariance of the trace distance
    
    \noindent\textit{Step $2$:} Assume true for $L=k$.
    \begin{align*}
        &\lVert\rho(x)-\rho(x')\rVert_1\\&=\lVert\prod_{l=1}^k\big[V_lU(\bold{x})\big]|\bold{0}\rangle\prod_{m=1}^k\langle\bold{0}|\big[U^\dagger(\bold{x})V^\dagger_m\big]\\&\hspace{0.5cm}-\prod_{l=1}^k\big[V_lU(\bold{x}')\big]|\bold{0}\rangle\prod_{m=1}^k\langle\bold{0}|\big[U^\dagger(\bold{x}')V^\dagger_m\big]\rVert_1\\&\leq k\lVert \mathcal{U}(\bold{x})-\mathcal{U}(\bold{x}')\rVert_\diamond
    \end{align*}
    \textit{Step $3$:}Prove true for $L=k+1$
    \begin{align*}
        &\lVert\rho(x)-\rho(x')\rVert_1\\&=\lVert\prod_{l=1}^{k+1}\big[V_lU(\bold{x})\big]|\bold{0}\rangle\prod_{m=1}^{k+1}\langle\bold{0}|\big[U^\dagger(\bold{x})V^\dagger_m\big]\\&-\prod_{l=1}^{k+1}\big[V_lU(\bold{x}')\big]|\bold{0}\rangle\prod_{m=1}^{k+1}\langle\bold{0}|\big[U^\dagger(\bold{x}')V^\dagger_m\big]\rVert_1\\&\stackrel{(a)}{=}\lVert U(\bold{x})\prod_{l=1}^{k}\big[V_lU(\bold{x})\big]|\bold{0}\rangle\langle\bold{0}|\prod_{m=1}^{k}\big[U^\dagger(\bold{x})V^\dagger_m\big]U^\dagger(\bold{x})\\&-U(\bold{x}')\prod_{l=1}^{k}\big[V_lU(\bold{x}')\big]|\bold{0}\rangle\langle\bold{0}|\prod_{m=1}^{k}\big[U^\dagger(\bold{x}')V^\dagger_m\big]U^\dagger(\bold{x}')\rVert_1\\&\stackrel{(b)}{\leq}\lVert U(\bold{x})\prod_{l=1}^{k}\big[V_lU(\bold{x})\big]|\bold{0}\rangle\langle\bold{0}|\prod_{m=1}^{k}\big[U^\dagger(\bold{x})V^\dagger_m\big]U^\dagger(\bold{x})\\&-U(\bold{x}')\prod_{l=1}^{k}\big[V_lU(\bold{x})\big]|\bold{0}\rangle\langle\bold{0}|\prod_{m=1}^{k}\big[U^\dagger(\bold{x})V^\dagger_m\big]U^\dagger(\bold{x}')\rVert_1\\&+\lVert U(\bold{x}')\prod_{l=1}^{k}\big[V_lU(\bold{x})\big]|\bold{0}\rangle\langle\bold{0}|\prod_{m=1}^{k}\big[U^\dagger(\bold{x})V^\dagger_m\big]U^\dagger(\bold{x}')\\&-U(\bold{x}')\prod_{l=1}^{k}\big[V_lU(\bold{x}')\big]|\bold{0}\rangle\langle\bold{0}|\prod_{m=1}^{k}\big[U^\dagger(\bold{x}')V^\dagger_m\big]U^\dagger(\bold{x}')\rVert_1\\&\stackrel{(c)}{\leq}\lVert \mathcal{U}(\bold{x})-\mathcal{U}(\bold{x}')\rVert_\diamond\\&+\Biggl[\lVert\prod_{l=1}^{k}\big[V_lU(\bold{x})\big]|\bold{0}\rangle\langle\bold{0}|\prod_{m=1}^{k}\big[U^\dagger(\bold{x})V^\dagger_m\big]\\&\underbrace{-\prod_{l=1}^{k}\big[V_lU(\bold{x}')\big]|\bold{0}\rangle\langle\bold{0}|\prod_{m=1}^{k}\big[U^\dagger(\bold{x}')V^\dagger_m\big]\rVert_1\Biggr]}_{\leq k\lVert \mathcal{U}(\bold{x})-\mathcal{U}(\bold{x}')\rVert_\diamond}\\&\leq (k+1)\lVert \mathcal{U}(\bold{x})-\mathcal{U}(\bold{x}')\rVert_\diamond
    \end{align*}
    where $(a)$ follows from unitary invariance of the trace norm and (b) from adding and subtracting $$U(\bold{x}')\prod_{l=1}^{k}\big[V_lU(\bold{x})\big]|\bold{0}\rangle\langle\bold{0}|\prod_{m=1}^{k}\big[U^\dagger(\bold{x})V^\dagger_m\big]U^\dagger(\bold{x}')$$ and using norm sub-additivity. Finally $(c)$ follows from the definition of the diamond norm, and the unitary invariance of the trace norm. 
    
    Since the initial condition and the inductive step are satisfied, this concludes the proof.
\end{proof}
We now proceed to bound the diamond norm $\lVert \mathcal{U}(\bold{x})-\mathcal{U}(\bold{x}')\rVert_\diamond$ for the angle embedding $|\psi(\xbf)\rangle$. In particular we use Lemma B.5 in \citep{caro2022generalization}: $\lVert \mathcal{U}(\bold{x})-\mathcal{U}(\bold{x}')\rVert_\diamond\leq2\max_{|\psi\rangle}\lVert \big[U(\xbf)-U(\xbf')\big]|\psi\rangle\rVert_2$, where we have 
\begin{align*}
&\max_{|\psi\rangle}\lVert \big[U(\xbf)-U(\xbf')\big]|\psi\rangle\rVert_2\\&=\Big[\langle\psi|\bigotimes_{j=1}^d\big(I-e^{i\delta \xbf_j\sigma_y}\big)e^{i \xbf_j\sigma_y}e^{-i \xbf_j\sigma_y}\big(I-e^{-i\delta \xbf_j\sigma_y}\big)|\psi\rangle\Big]^{1/2}\\
    &=\Big[\langle\psi|\bigotimes_{j=1}^d2I\big(1-\cos(\delta \xbf_j)\big)|\psi\rangle\Big]^{1/2}\\&\leq\Big[\langle\psi|\bigotimes_{j=1}^d\delta \xbf_j^2I|\psi\rangle\Big]^{1/2}\\&=\prod_{i=1}^d|\delta \xbf_j|\stackrel{(a)}{\leq}\Big(\frac{\lVert\delta \bold{x}\rVert_1}{d}\Big)^d \leq\epsilon^dd^{-d/p},
\end{align*}
where $I$ is the identity matrix and $(a)$ follows from the AM-GM inequality.

Combining everything together we have that 
\begin{align}
    \mathcal{S}_{r,p,\epsilon}^C\leq 2L\epsilon^dd^{-d/p}
\end{align}
Finally, we bound the diamond norm of the dense embedding again using Lemma B.5 in \citep{caro2022generalization} as
\begin{align*}
     &\max_{|\psi\rangle}\lVert \big[U(x)-U(\xbf')\big]|\psi\rangle\rVert_2\\&=\Big[\langle\psi|\bigotimes_{j=1}^{d/2}\big(I-e^{i\delta \xbf_{2j-1}\sigma_y}e^{i\delta \xbf_{2j}\sigma_z}\big)\big(I-e^{-i\delta \xbf_{2j}\sigma_z}e^{-i\delta \xbf_{2j-1}\sigma_y}\big)|\psi\rangle\Big]^{0.5}\\&=\Big[\langle\psi|\bigotimes_{j=1}^{d/2}2I\big(1-\cos(\delta \xbf_{2j-1})\cos(\delta \xbf_{2j})\big)|\psi\rangle\Big]^{0.5}\\&\leq\Big[\langle\psi|\bigotimes_{j=1}^{d/2}(\delta \xbf_{2j-1}^2+\delta \xbf_{2j}^2)I|\psi\rangle\Big]^{1/2}\\&=\Big[\prod_{j=1}^{d/2}(\delta \xbf_{2j-1}^2+\delta \xbf_{2j}^2)\Big]^{1/2}\\&\stackrel{(a)}{\leq}\Big(\frac{2\lVert\delta\bold{x}\rVert_2^2}{d}\Big)^{d/4}\\&\leq\big(2\epsilon^2\max\{d^{-1},d^{-2/p}\}\big)^{d/4},
\end{align*}
where $(a)$ follows again from the AM-GM inequality. Combining everything together we get
\begin{align}
    \mathcal{S}^c_{r,p,\epsilon}\leq 2L(\sqrt{2}\epsilon)^{d/2}\max\{d^{-d/4},d^{-d/2p}\}
\end{align}

\subsection{Quantum attacks}
The adversarial examples of quantum attacks are defined as $\rho'=\underset{\rho^{*}:\lVert\rho^{*}-\rho\rVert_p\leq\epsilon}{\arg\min}y{\rm Tr}\big(O\rho^{*}\big)$. With this definition we can upper bound the term $\mathcal{S}^Q_{r,p,\epsilon}$ as follows
\begin{align*}
    \mathcal{S}^Q_{r,p,\epsilon}&=\max_{\rho':\lVert\rho-\rho'\rVert_p\leq\epsilon}\lVert\rho-\rho'\rVert_{r/(r-1)}\\&\leq\begin{cases}
        \lVert\rho-\rho'\rVert_p&\;\text{if}\;1-1/r-1/p\leq0\\\lVert\rho-\rho'\rVert_pd_\mathcal{H}^{1-1/r-1/p}&\;\text{if}\;1-1/r-1/p\geq0       
    \end{cases}
\end{align*}
This is an application of H{\"o}lder's inequality in the case when $\frac{r}{r-1}\geq p$ and otherwise an application of $\lVert A\rVert_r\geq\lVert A\rVert_s$ if $r\leq s$. Thus 
\begin{align*}  \mathcal{S}^Q_{r,p,\epsilon}\leq\epsilon\max\{1, d_\mathcal{H}^{1-1/q-1/p}\}
\end{align*}
\section{Proof of Theorem~\ref{thm:lowerbound}}\label{app:lower}
We first derive the lower bound.
\subsection{Lower bound}
We begin by considering the standard RC of the function class $\mathcal{F}_2$. 
\begin{align*}
    \mathcal{R}(\mathcal{F}_r)&=\frac{1}{m}\underset{\boldsymbol{\sigma}}{\mathbb{E}}\sup_{A\in \mathcal{A}_r}\Biggl[\sum_{i=1}^m\sigma_i y_i{\rm Tr}\big(A\rho(\xbf_i)\big)\Biggr]\\&=\frac{1}{m}\underset{\sbf}{\mathbb{E}}\Biggl[\sup_{A\in\mathcal{A}_r}{\rm Tr}\Big(A\sum_{i=1}^m\sigma_i y_i\rho(\xbf_i)\Big)\Biggr]\\&\stackrel{(a)}{=}\frac{b}{m}\underset{\boldsymbol{\sigma}}{\mathbb{E}}\Biggl[\Big\lVert\sum_{i=1}^m\sigma_i y_i\rho(\xbf_i)\Big\rVert_{\frac{r}{r-1}}\Biggr]
\end{align*}
where $(a)$ follows from H{\"o}lder's inequality, with equality guaranteed by the supremum operation. We define the supremum achieving observable as 
\begin{align*}
    A_{\sbf}^{*}=\underset{A\in {\mathcal{A}_r}}{\arg\sup}\sum_{i=1}^m\sigma_iy_i{\rm Tr}\big(A\rho(\xbf_i)\big)
\end{align*}
With this, the RC may be written as 
\begin{align*}
    \mathcal{R}(\mathcal{F}_r)=\frac{1}{m}\underset{\sbf}{\mathbb{E}}\Biggl[{\rm Tr}\Big(A_{\sbf}^{*}\sum_{i=1}^m\sigma_iy_i\rho(\xbf_i)\Big)\Biggr]
\end{align*}
We note that the RC is invariant under the transformation $\sbf\mapsto-\sbf$. Therefore we have that
\begin{align}
    \sbf\mapsto-\sbf\implies A_{\sbf}^{*}\mapsto A_{-\sbf}^{*}=-A_{\sbf}^{*},\label{eq:opt_obs}
\end{align} 
which leaves the product $\sbf A_{\sbf}^{*}$ invariant as required.

We now switch our focus to the ARC $\mathcal{R}(\mathcal{F}_{r,p,\epsilon})$. We use the notation $A_{\sbf}^{*}$ to denote the supremum achieving observable of the \textbf{standard} RC, as before.
\begin{align*}
    &\mathcal{R}(\mathcal{F}_{r,p,\epsilon})\\&=\frac{1}{m}\underset{\sbf}{\mathbb{E}}\Bigl[\sup_{A\in\mathcal{A}_r}\sum_{i=1}^m\sigma_i\min_{\rho':\lVert\rho'-\rho(\xbf_i)\rVert_p\leq\epsilon}y_i{\rm Tr}\big(O\rho'\big)\Bigr]\\&\stackrel{(a)}{\geq}\frac{1}{m}\underset{\sbf}{\mathbb{E}}\Bigl[\sum_{i=1}^m\sigma_i\min_{\rho':\lVert\rho'-\rho(\xbf_i)\rVert_p\leq\epsilon}y_i{\rm Tr}\big(A_{\sbf}^{*}\rho'\big)\Bigr]\\&=\frac{1}{m}\underset{\sbf}{\mathbb{E}}\Bigl[\sum_{i=1}^m\sigma_i\min_{\tau:\lVert\tau\rVert_p\leq\epsilon,\;{\rm Tr}\tau=0, \rho(\xbf_i)+\tau\succeq0}y_i{\rm Tr}\big(A_{\sbf}^{*}(\rho(\xbf_i)+\tau)\big)\Bigr]\\&=\mathcal{R}(\mathcal{F})+\frac{1}{m}\underset{\sbf}{\mathbb{E}}\Bigl[\sum_{i=1}^m\sigma_i\min_{\tau:\lVert\tau\rVert_p\leq\epsilon,\;{\rm Tr}\tau=0, \;\rho(\xbf_i)+\tau\succeq0}y_i{\rm Tr}\big(A_{\sbf}^{*}\tau\big)\Bigr]
\end{align*}
where $(a)$ follows from replacing the supremum over $\Acal_r$ with the specific observable $A_{\sbf}^{*}$.

We now analyze the term $$\sigma_i\min_{\tau:\lVert\tau\rVert_p\leq\epsilon,\;{\rm Tr}\tau=0, \;\rho(\xbf_i)+\tau\succeq0}y_i{\rm Tr}\big(A_{\boldsymbol{\sigma}^{*}}\tau\big).$$ In particular we note that under Assumption~\ref{assum:1}, the following equivalence relation holds
\begin{align*}
    &\{\tau:\lVert\tau\rVert_p\leq\epsilon,\;{\rm Tr}\tau=0, \;\rho_i+\tau\succeq0\}\\&=\{\tau:\lVert\tau\rVert_p\leq\epsilon,\;{\rm Tr}\tau=0\}
\end{align*}
And therefore minimization over the two sets is equivalent. Let us define 
\begin{align*}
    \tau_{\sbf,i}^{*}=\underset{\tau:\lVert\tau\rVert_p\leq\epsilon,\;{\rm Tr}\tau=0}{\arg\min}y_i{\rm Tr}(A_{\sbf}^{*}\tau)
\end{align*}
Under $\sbf\mapsto-\sbf$ we have
\begin{align}
    \tau_{\sbf,i}^{*}\mapsto\tau_{-\sbf,i}^{*}&=\underset{\tau:\lVert\tau\rVert_p\leq\epsilon,\;{\rm Tr}\tau=0}{\arg\min}y_n{\rm Tr}(A_{-\sbf}^{*}\tau)\nonumber\\&=\underset{\tau:\lVert\tau\rVert_p\leq\epsilon,\;{\rm Tr}\tau=0}{\arg\min}y_i{\rm Tr}(-A_{\sbf}^{*}\tau)\nonumber\\&=\underset{-\tau:\lVert-\tau\rVert_p\leq\epsilon,\;{\rm Tr}(-\tau)=0}{\arg\min}y_i{\rm Tr}(A_{\sbf}^{*}\tau)\nonumber\\&=\underset{-\tau:\lVert\tau\rVert_p\leq\epsilon,\;{\rm Tr}\tau=0}{\arg\min}y_i{\rm Tr}(A_{\sbf}^{*}\tau)\nonumber\\&=-\underset{\tau:\lVert\tau\rVert_p\leq\epsilon,\;{\rm Tr}\tau=0}{\arg\min}y_i{\rm Tr}(A_{\sbf}^{*}\tau)=-\tau_{\sbf,i}^{*} \label{eq:opt_pert}
\end{align}
Combining the above results from equations \ref{eq:opt_obs}, \ref{eq:opt_pert}, we have \begin{align*}\sbf\mapsto-\sbf&\implies\sigma_i y_i{\rm Tr}\big(A_{\sbf}^{*}\tau_{\sbf,i}^{*}\big)\mapsto-\sigma_i y_i{\rm Tr}\big(A_{-\sbf}^{*}\tau_{-\sbf,i}^{*}\big)\\&\hspace{0.2cm}\stackrel{(a)}{=}-\sigma_i y_i{\rm Tr}\big(A_{\sbf}^{*}\tau_{\sbf,i}^{*}\big)\end{align*}where $(a)$ follows from direct application of \ref{eq:opt_obs}, \ref{eq:opt_pert}.
It follows that 
\begin{align*}
    \underset{\sbf}{\mathbb{E}}\Bigl[\sbf_i y_i{\rm Tr}\big(A_{\sigma}^{*}\tau_{\sbf,i}^{*}\big)\Bigr]=0
\end{align*} as the vectors $\sbf$ and $-\sbf$ are equiprobable.

Going back to the ARC lower bound, this implies
\begin{align*}
    &\mathcal{R}(\mathcal{F}_{r,p,\epsilon})\\&\geq\mathcal{R}(\mathcal{F}_r)+\frac{1}{m}\underset{\sbf}{\mathbb{E}}\Biggl[\sum_{i=1}^m\sigma_i\min_{\tau:\lVert\tau\rVert_p\leq\epsilon,\;{\rm Tr}\tau=0, \;\rho(\xbf_i)+\tau\succeq0}y_i{\rm Tr}\big(A_{\sbf}^{*}\tau\big)\Biggr]\\&=\mathcal{R}(\mathcal{F}_r)+\frac{1}{m}\underset{\sbf}{\mathbb{E}}\Biggl[\sum_{i=1}^m\sigma_iy_i{\rm Tr}\big(A_{\sbf}^{*}\tau_{\sbf,i}^{*}\big)\Biggr]\\&=\mathcal{R}(\mathcal{F}_r).
\end{align*}
\subsection{Upper bound}
To upper bound the ARC we will make use of the sub-additivity of the supremum operation. In particular
\begin{align}
    &\mathcal{R}(\mathcal{F}^Q_{r,p,\epsilon})\\&=\frac{1}{m}\underset{\sbf}{\mathbb{E}}\Bigl[\sup_{A\in\mathcal{A}_r}\sum_{i=1}^m\sigma_i\min_{\rho':\lVert\rho'-\rho(\xbf_i)\rVert_p\leq\epsilon}y_i{\rm Tr}\big(A\rho'\big)\Bigr]\\&=\frac{1}{m}\underset{\sbf}{\mathbb{E}}\Bigl[\sup_{A\in\mathcal{A}_r}\sum_{i=1}^m\sigma_i\min_{\tau:\begin{cases}\lVert\tau\rVert_p\leq\epsilon\\{\rm Tr}\;\tau=0\\\rho(\xbf_i)+\tau\succeq0\end{cases}}y_i{\rm Tr}\big(A(\rho(\xbf_i)+\tau)\big)\Bigr]\\&\stackrel{(a)}{\leq}\mathcal{R}(\mathcal{F})+\underbrace{\frac{1}{m}\underset{\sbf}{\mathbb{E}}\Bigl[\sup_{A\in\mathcal{A}_r}\sum_{i=1}^m\sigma_i\min_{\tau:\begin{cases}\lVert\tau\rVert_p\leq\epsilon\\{\rm Tr}\;\tau=0\\\rho(\xbf_i)+\tau\succeq0\end{cases}}y_i{\rm Tr}\big(A(\tau)\big)}_{\coloneqq\mathcal{R}(\Tilde{\mathcal{F}}_{r,p,\epsilon})}\Bigr]\label{eq:noisy_upper}
\end{align}
where $(a)$ follows from the sub-additivity of the supremum and the definition of the standard RC.

As before, we note that under Assumption \ref{assum:1}, the following equivalence relation holds
\begin{align*}
    &\{\tau:\lVert\tau\rVert_p\leq\epsilon,\;{\rm Tr}\tau=0, \;\rho(\xbf_i)+\tau\succeq0\}\\&=\{\tau:\lVert\tau\rVert_p\leq\epsilon,\;{\rm Tr}\tau=0\}.
\end{align*}
Let us define 
\begin{align*}
    \tau_y^{*}=\underset{\tau:\lVert\tau\rVert_p\leq\epsilon,\;{\rm Tr}\tau=0}{\arg\min}y{\rm Tr}(A\tau).
\end{align*}
Under $y\mapsto-y$ we have
\begin{align}
    \tau_y^{*}\mapsto\tau_{-y}^{*}&=\underset{\tau:\lVert\tau\rVert_p\leq\epsilon,\;{\rm Tr}\tau=0}{\arg\min}-y{\rm Tr}(O\tau)\\&=\underset{-\tau:\lVert-\tau\rVert_p\leq\epsilon,\;{\rm Tr}(-\tau)=0}{\arg\min}y{\rm Tr}(O\tau)\\&=\underset{-\tau:\lVert\tau\rVert_p\leq\epsilon,\;{\rm Tr}\tau=0}{\arg\min}y{\rm Tr}(O\tau)\\&=-\underset{\tau:\lVert\tau\rVert_p\leq\epsilon,\;{\rm Tr}\tau=0}{\arg\min}y{\rm Tr}(O\tau)=-\tau_{y}^{*}\label{eq:opt_pert_upper}
\end{align}
Thus, we have
\begin{align*}
    \mathcal{R}(\Tilde{\mathcal{F}}_{r,p,\epsilon})&=\frac{1}{m}\underset{\sbf}{\mathbb{E}}\Bigl[\sup_{A\in\mathcal{A}_r}\sum_{i=1}^m\sigma_i\min_{\tau:\lVert\tau\rVert_p\leq\epsilon,\;{\rm Tr}\tau=0}y_i{\rm Tr}\big(A\tau\big)\Bigr]\\&=\frac{1}{m}\underset{\sbf}{\mathbb{E}}\Big[\sup_{A\in\mathcal{A}_r}\sum_{i=1}^m\sigma_i{\rm Tr}\big(A\tau_{+1}^{*}\big)\Bigr]\\&=\frac{1}{m}\underset{\sbf}{\mathbb{E}}\Bigl[\sup_{A\in\mathcal{A}_r}{\rm Tr}\Big(\sum_{i=1}^m\sigma_iA\tau_{+1}^{*}\Big)\Bigr]\\
    &\leq\frac{1}{m}\underset{\sbf}{\mathbb{E}}\Bigl[\sup_{A\in\mathcal{A}_r}\Big\lVert A\Big\rVert_r\Big\lVert\sum_{i=1}^m\sigma_i\tau_{+1}^{*}\Big\rVert_{\frac{r}{r-1}}\Bigr]\\&\leq\frac{b}{m}\underset{\sbf}{\mathbb{E}}\Bigl[\Big\lVert\sum_{i=1}^m\sigma_i\tau_{+1}^{*}\Big\rVert_{\frac{r}{r-1}}\Bigr]\\&\leq\frac{b\epsilon\max\{1, d^{1-1/p-1/r}\}}{\sqrt{m}},
\end{align*}
where the last step is an application of Khintchine's inequality. Using this in equation \ref{eq:noisy_upper}, we have 
\begin{align*}
    \mathcal{R}(\mathcal{F}^Q_{r,p,\epsilon})\leq\mathcal{R}(\mathcal{F})+\frac{b\epsilon\max\{1, d^{1-1/p-1/r}\}}{\sqrt{m}}
\end{align*}
\section{Rademacher Complexity of Multi-Class Quantum Classifier}\label{app:multi-class}
In this section, we derive upper bounds on the non-adversarial and adversarial Rademacher complexity for multi-class quantum classifiers. 
\subsection{Non-Adversarial Multiclass Rademacher Complexity}
To obtain an upper bound on $\Rcal(\ell \circ \Gcal_r)$, note that
\begin{align*}
&\Rcal(\ell \circ \Gcal_r) \\&=\Ebb_{\boldsymbol{\sigma}}\Bigl[\sup_{f_k \in \Fcal_r, \forall k \in [K]} \frac{1}{m} \sum_{i=1}^m \sigma_i  \phi_{\gamma} \Bigl(\min_{k \neq y_i} f_{y_i}(\xbf) -f_k(\xbf_i)\Bigr)\Bigr] \nonumber.
\end{align*}
Following the proof of \citep[Theorem 9.2]{mohri2012foundations}, we then get that 
\begin{align*}
  \Rcal(\ell \circ \Gcal_r) &\leq \frac{2K}{\gamma} \Ebb_{\boldsymbol{\sigma}}\Bigl [\sup_{k \in [K],f_k \in \Fcal_r} \frac{1}{m} \sum_{i=1}^m \sigma_i f_{k}(\xbf_i)\Bigr] 
\\ &=\frac{2K}{\gamma}\Ebb_{\boldsymbol{\sigma}}\Bigl[ \sup_{y \in \Ycal, A_y \in \Acal_r} \frac{1}{m} \sum_{i=1}^m  \sigma_i \Tr(\rho(\xbf_i) A_y)\Bigr]\nonumber \\
    &= \frac{2K}{m\gamma} \Ebb_{\boldsymbol{\sigma}}\Bigl[ \sup_{y \in \Ycal, A_y \in \Acal_r} \Tr\Bigl((\sum_{i=1}^m  \sigma_i \rho(\xbf_i))A_y \Bigr)\Bigr]\nonumber \\
    & \leq \frac{2K}{m \gamma} \Ebb_{\boldsymbol{\sigma}}\Bigl[\sup_{y \in \Ycal, A_y \in \Acal_r} \Vert A_y \Vert_r \Bigl\Vert \sum_{i=1}^m  \sigma_i \rho(\xbf_i)\Bigr\Vert_{\frac{r}{r-1}} \Bigr] \nonumber \\
    & \leq \frac{2Kb}{m\gamma}\Ebb_{\boldsymbol{\sigma}}\Bigl[\Bigl\Vert \sum_{i=1}^m  \sigma_i \rho(\xbf_i)\Bigr\Vert_{\frac{r}{r-1}} \Bigr],
\end{align*} where the last term can be upper bounded as in the proof of Theorem~\ref{thm:binary_nonadversarial}. 
\subsection{Adversarial Rademacher Complexity Analysis}
We derive upper bound on the adversarial Rademacher complexity for multi-class classification defined as
\begin{align}
    &\Rcal(\ell^C_{r,p,\epsilon} \circ \Gcal_r)\nonumber\\ &=\Ebb_{\boldsymbol{\sigma}}\Bigl[\sup_{\fbf \in \Fcal_r} \frac{1}{m} \sum_{i=1}^m \sigma_i \max_{\xbf': \Vert \xbf_i-\xbf'\Vert_p \leq \epsilon} \phi_{\gamma} \Bigl( M(\fbf(\xbf'),y_i)\Bigr)\Bigr] \nonumber\\
    &= \Ebb_{\boldsymbol{\sigma}}\Bigl[\sup_{\fbf \in \Fcal_r} \frac{1}{m} \sum_{i=1}^m \sigma_i  \phi_{\gamma} \Bigl( \min_{\xbf': \Vert \xbf_i-\xbf'\Vert_p \leq \epsilon} \min_{k \neq y_i} f_{y_i}(\xbf')-f_k(\xbf')\Bigr)\Bigr] \nonumber \\
    &= \Ebb_{\boldsymbol{\sigma}}\Bigl[\sup_{\fbf \in \Fcal_r} \frac{1}{m} \sum_{i=1}^m \sigma_i  \phi_{\gamma} \Bigl( \min_{k \neq y_i} \min_{\xbf': \Vert \xbf_i-\xbf'\Vert_p \leq \epsilon} f_{y_i}(\xbf')-f_k(\xbf')\Bigr)\Bigr] \label{eq:multiclass_proofeq_1}
\end{align} It can then be easily verified that (see Proof of Theorem 8 in \citep{xiao2022adversarial}), 
\begin{align*}
    \phi_{\gamma} \Bigl(\min_{k \neq y_i} \min_{\xbf': \Vert \xbf_i-\xbf
    '\Vert_p \leq \epsilon}  f_{y_i}(\xbf')-f_k(\xbf')\Bigr) = \phi_{\gamma}(\min_{k} h^k(\xbf_i,y_i))
\end{align*} where 
\begin{align}
    h^k(\xbf_i,y_i)=\min_{\xbf': \Vert \xbf_i-\xbf'\Vert_p \leq \epsilon} f_{y_i}(\xbf')-f_k(\xbf') + \gamma \Ibb\{k=y_i\}.
\end{align} We now proceed as in the proof of Theorem 9.2 in \citep{mohri2012foundations} to get that
\begin{align}
    &\Rcal(\ell_{r,p,\epsilon} \circ \Gcal_r) \nonumber\\&=\Ebb_{\boldsymbol{\sigma}}\Bigl[\sup_{\fbf \in \Fcal_r} \frac{1}{m} \sum_{i=1}^m \sigma_i  \phi_{\gamma} \Bigl( \min_{k} h^k(\xbf_i,y_i)\Bigr)\Bigr] \nonumber \\
    & = \Ebb_{\boldsymbol{\sigma}}\Bigl[\sup_{\fbf \in \Fcal_r} \frac{1}{m} \sum_{i=1}^m \sigma_i \max_k \phi_{\gamma} \Bigl(  h^k(\xbf_i,y_i)\Bigr)\Bigr] \nonumber\\
    & \stackrel{(a)}{\leq} \sum_{k=1}^K \Ebb_{\boldsymbol{\sigma}}\Bigl[\sup_{\fbf \in \Fcal_r} \frac{1}{m} \sum_{i=1}^m \sigma_i \phi_{\gamma} \Bigl(  h^k(\xbf_i,y_i)\Bigr)\Bigr] \nonumber \\
    & \stackrel{(b)}{\leq}\frac{1}{\gamma}\sum_{k=1}^K \Ebb_{\boldsymbol{\sigma}}\Bigl[\sup_{\fbf \in \Fcal_r} \frac{1}{m} \sum_{i=1}^m \sigma_i   h^k(\xbf_i,y_i)\Bigr] \nonumber\\
    & {\leq} \frac{1}{\gamma}\sum_{k=1}^K \Ebb_{\boldsymbol{\sigma}}\Bigl[\sup_{\fbf \in \Fcal_r} \frac{1}{m} \sum_{i=1}^m \sigma_i   \min_{\xbf': \Vert \xbf_i-\xbf'\Vert_p \leq \epsilon} f_{y_i}(\xbf')-f_k(\xbf')\Bigr] \nonumber\\&+\frac{1}{\gamma}\sum_{k=1}^K \Ebb_{\boldsymbol{\sigma}}\Bigl[\sup_{\fbf \in \Fcal_r} \frac{1}{m} \sum_{i=1}^m \sigma_i   \gamma \Ibb\{k=y_i\} \Bigr] \nonumber \\
    &= \frac{1}{\gamma}\sum_{k=1}^K \Ebb_{\boldsymbol{\sigma}}\Bigl[\sup_{\fbf \in \Fcal_r} \frac{1}{m} \sum_{i=1}^m \sigma_i   \min_{\xbf': \Vert \xbf_i-\xbf'\Vert_p \leq \epsilon} f_{y_i}(\xbf')-f_k(\xbf')\Bigr] \nonumber \end{align} \begin{align}
    & \stackrel{(c)}{\leq}  \frac{1}{\gamma}\sum_{k=1}^K \Ebb_{\boldsymbol{\sigma}}\Bigl[\sup_{\fbf \in \Fcal_r} \frac{1}{m} \sum_{i=1}^m \sigma_i   f_{y_i}(\xbf_i)-f_k(\xbf_i)\Bigr] \nonumber \\&+\frac{1}{\gamma}\sum_{k=1}^K \Ebb_{\boldsymbol{\sigma}}\Bigl[\sup_{\fbf \in \Fcal_r} \frac{1}{m} \sum_{i=1}^m \sigma_i   \nonumber\\&\hspace{0.4cm}\min_{\xbf': \Vert \xbf_i-\xbf'\Vert_p \leq \epsilon} (f_{y_i}(\xbf')-f_k(\xbf'))-(f_{y_i}(\xbf_i)-f_k(\xbf_i))\Bigr] \\
    &\stackrel{(d)}{\leq} \frac{2 K}{\gamma} \Rcal(\Pi(\Fcal_r)) \nonumber\\&+\frac{1}{\gamma}\sum_{k=1}^K \Ebb_{\boldsymbol{\sigma}}\Bigl[\sup_{\fbf \in \Fcal_r} \frac{1}{m} \sum_{i=1}^m \sigma_i   \nonumber\\&\hspace{0.4cm}\underbrace{\min_{\xbf': \Vert \xbf_i-\xbf'\Vert_p \leq \epsilon} (f_{y_i}(\xbf')-f_k(\xbf'))-(f_{y_i}(\xbf_i)-f_k(\xbf_i))}_{:=g_k(\xbf_i,y_i)}\Bigr]
\end{align} where the inequality in $(a)$ follows from Lemma 9.1 in \citep{mohri2012foundations}, and the inequality in $(b)$ follows from the contraction inequality. The inequality in $(c)$ follows by adding and subracting the term $f_{y_i}(\xbf_i)-f_k\xbf_i)$ and then using $\sup(a+b) \leq \sup(a)+\sup(b)$. Finally, the inequality in $(d)$ follows from \citep[Theorem 9.2]{mohri2012foundations}.

We now upper bound the term,
\begin{align}
  &\Ebb_{\boldsymbol{\sigma}}\Bigl[\sup_{\fbf \in \Fcal_r} \frac{1}{m} \sum_{i=1}^m \sigma_i g_k(\xbf_i,y_i) \Bigr] \nonumber\\ &= \Ebb_{\boldsymbol{\sigma}}\Bigl[\sup_{A_j \in \Acal_r, j\in [K]} \frac{1}{m} \sum_{i=1}^m \sigma_i  \nonumber\\&\hspace{0.9cm}\min_{x': \Vert \xbf-\xbf'\Vert_p \leq \epsilon} \Tr\Bigl((\rho(\xbf')-\rho(\xbf_i))(A_{y_i}-A_k)\Bigr)\Bigr]. \label{eq: intermediate_1.}  
\end{align} To upper bound the Rademacher complexity in \eqref{eq: intermediate_1.}, we resort to the covering number based argument as before. To this end, we note that the function $g_k(\xbf,y)$ is determined by the observables $A_{y}$ and $A_k$. Let us now consider the function $g^c_k(\xbf,y)$ determined by the observables $A^c_{y}$ and $A^c_k$ where $A^c_{y}, A^c_k \in \Ccal(\Acal_r, \Vert \cdot \Vert_r, \delta')$ belongs to the $\delta'$-cover of $\Acal$ with respect to the $r$-Schatten norm. Consequently, we get that
\begin{align*}
    &g_k(\xbf,y) -g^c_k(\xbf,y)|\\&= \Bigl|\min_{\xbf': \Vert \xbf-\xbf'\Vert_p \leq \epsilon} \Tr\Bigl((\rho(\xbf')-\rho(\xbf))(A_{y}-A_k)\Bigr)\\&\hspace{0.2cm}-\min_{\xbf': \Vert \xbf-\xbf
    '\Vert_p \leq \epsilon} \Tr\Bigl((\rho(\xbf')-\rho(\xbf))(A^c_{y}-A^c_k)\Bigr) \Bigr| \\
    & = \Bigl| \Tr\Bigl((\rho(\xbf^*)-\rho(\xbf))(A_{y}-A_k)\Bigr)\\&\hspace{0.2cm}- \Tr\Bigl((\rho(\xbf^{*c})-\rho(\xbf_i))(A^c_{y}-A^c_k)\Bigr) \Bigr|\\
    & \leq \Bigl| \Tr\Bigl((\rho(\xbf^{*c})-\rho(\xbf))(A_y-A_k)\Bigr)\\&\hspace{0.2cm}- \Tr\Bigl((\rho(\xbf^{*c})-\rho(\xbf))(A^c_y-A^c_k)\Bigr) \Bigr|
\end{align*} 
\begin{align*}
    & = \Bigl| \Tr\Bigl((\rho(\xbf^{*c})-\rho(\xbf))(A_{y}-A_k-(A^c_{y}-A^c_k))\Bigr) \Bigr|\\
    & \leq \Vert A_{y}-A_k-(A^c_{y}-A^c_k) \Vert_r  \Vert \rho(\xbf^{*c})-\rho(\xbf)\Vert_{\frac{r}{r-1}}\\
    & \leq \mathcal{S}^{Q/C}_{r,p,\epsilon} (\Vert A_y -A^c_y \Vert_r+\Vert A_k -A^c_k \Vert_r)\\& \leq 2\delta' \mathcal{S}^{Q/C}_{r,p,\epsilon}.
\end{align*}
This gives that the set $$ \{(\xbf,y) \mapsto g_k(\xbf,y): A_k, A_y \in \Ccal(\Acal_r, \Vert \cdot \Vert_r, \delta/(2\mathcal{S}_{r,p,\epsilon}^{Q/C}))\}$$ indeed results in a $\delta$-cover of the set   $ \tilde{\Fcal}^k_{r,p,\epsilon}=\{(\xbf,y) \mapsto g_k(\xbf,y): A_k, A_y \in \Acal_r\}$. In particular, a $\delta/\sqrt{m}$-cover of $\tilde{\Fcal}^k_{r,p,\epsilon}$ implies a $\delta$-cover of the the set $\tilde{\Fcal}^k_{r,p,\epsilon}(\Dcal)=\{\mathbf{g}_k=(g_k(\xbf_1,y_1),\hdots, g_k(\xbf_m,y_m)): g_k(\xbf,y) \in \tilde{\Fcal}_{r,p,\epsilon})$ according to the $l_2$-norm.

Subsequently, using Dudley entropy integral bound then yields that
\begin{align}
& \Ebb_{\boldsymbol{\sigma}}\Bigl[\sup_{\fbf \in \Fcal_r} \frac{1}{m} \sum_{i=1}^m \sigma_i g_k(\xbf_i,y_i) \Bigr] \nonumber \\& \leq 12 \int_0^{D/2} \frac{\sqrt{\mathcal{N}(\tilde{\Fcal}^k_{r,p,\epsilon}(\Dcal), \Vert \cdot \Vert_2, \alpha))}}{m} d\alpha \label{eq:dudley_integral_multiclass}
\end{align}
with 
\begin{align*}&D= \max_{g_k \in \tilde{\Fcal}_{p,\epsilon}^k} \sqrt{\sum_{i=1}^m g_k(\xbf_i,y_i)^2} \\&= \max_{g_k} \sqrt{\sum_{i=1}^m \Bigl(\Tr\Bigl((\rho(\xbf^*_i)-\rho(\xbf_i))(A_{y_i}-A_k)\Bigr)\Bigr)^2}\\&\leq \max_{g_k} \sqrt{\sum_{i=1}^m (\mathcal{S}^{Q/C}_{r,p,\epsilon})^2 \Vert A_{y_i}-A_k \Vert_r^2} \\&\leq 2\sqrt{m}b\mathcal{S}_{r,p,\epsilon}^{Q/C}.
\end{align*}
Subsequently, evaluating the integral in \eqref{eq:dudley_integral_multiclass} then yields that
\begin{align} &\Ebb_{\boldsymbol{\sigma}}\Bigl[\sup_{\fbf \in \Fcal_r} \frac{1}{m} \sum_{i=1}^m \sigma_i g_k(\xbf_i,y_i) \Bigr] \nonumber\\& \leq 12  \int_0^{b\rob\sqrt{m}} \frac{\sqrt{d^2_H \log(\frac{3 \times 2^{1+1/r} \sqrt{m}\rob b}{\delta})}}{m}d\delta\nonumber\\&=\frac{72bd_H\rob2^{1/r}}{\sqrt{m}}\int_0^{2^{-1-1/r}/3}\log(\frac{1}{\delta})d\delta\nonumber\\&=\frac{72bd_H\rob2^{1/r}}{\sqrt{m}}\Biggl(\frac{2^{-1/r}}{6}\sqrt{\log(2^{1/r}6)}\nonumber\\ &\hspace{0.4cm}+\frac{\sqrt{\pi}}{2}\Big(1-{\rm erf}\big(\sqrt{\log(2^{1/r}6)}\big)\Big)\Biggl).
\end{align}
\section{Quantum attacks via quantum FGSM}\label{appdx:QFGSM}
In this appendix we describe our proposed quantum attack, i.e. an adversarial attack directly perturbing the quantum state $\rho(\xbf)$. It is implemented by a parametrized quantum circuit which maps $\rho(\xbf)\mapsto \rho'=\mathcal{U}_\theta(\rho(\xbf))$.

\begin{algorithm}[H]
\caption{Quantum FGSM}\label{alg:QPFGSM}
\begin{algorithmic}[1]
\Require Quantum classifier $g$, loss function $\ell(\cdot)$, data tuple $(\rho(\xbf), y)$
\Require Adversarial parameters $(p, \epsilon)$, maximum number of iterations $\mathtt{max\_iter}$
\Require Adversarial quantum channel $\mathcal{U}_\theta(\cdot)$ such that $\mathcal{U}_{\theta=\theta_0}(\cdot) = \mathbb{I}$, learning rate $\alpha$
\State $i \gets 0$
\State $\theta \gets \alpha \cdot \sign\left(\nabla_\theta \ell\big(g, \mathcal{U}_{\theta=\theta_0}(\rho(\xbf)), y\big)\right)$
\While{$\lVert \rho(\xbf) - \mathcal{U}_\theta(\rho(\xbf)) \rVert_p \geq \epsilon$ \textbf{and} $i < \mathtt{max\_iter}$}
    \State $\theta \gets \theta / 2$
    \State $i \gets i + 1$
\EndWhile
\If{$i = \mathtt{max\_iter}$}
    \State $\theta \gets \theta_0$
\EndIf
\State \Return $\mathcal{U}_\theta(\rho(\xbf))$
\end{algorithmic}
\end{algorithm}

\newpage
\bibliography{refs}
\end{document}